\documentclass[twocolumn]{article}
\usepackage{imr}
\usepackage{graphicx}
\def\thepage{}
\usepackage{subcaption}

\usepackage{amsmath}
\usepackage{amssymb}
\usepackage{amsthm}
\usepackage{placeins}
\usepackage{mathtools}

\usepackage{colortbl}
\usepackage{xcolor}
\definecolor{DarkRed}{rgb}{0.55,0,0}
\definecolor{DarkBlue}{rgb}{0,0,0.55}
\usepackage{algsimple}

\usepackage{tikz}
 \usetikzlibrary{arrows}
 \usetikzlibrary{arrows.meta}
\newtheorem{definition}{Definition}
\newtheorem{corollary}{Corollary}
\newtheorem{proposition}{Proposition}

\newcommand\powercell{P}

\begin{document}

\title{\uppercase{Higher-dimensional power diagrams for semi-discrete optimal transport}}
\author{Philip Claude Caplan}
\date{Middlebury College Department of Computer Science, Middlebury, VT, U.S.A.,\\pcaplan@middlebury.edu}

\abstract{
Efficient algorithms for solving optimal transport problems are important for measuring and optimizing distances between functions.
In the $L^2$ semi-discrete context, this problem consists of finding a map from a continuous density function to a discrete set of points so as to minimize the transport cost, using the squared Euclidean distance as the cost function.
This has important applications in image stippling, clustering, resource allocation and in generating blue noise point distributions for rendering.
Recent algorithms have been developed for solving the semi-discrete problem in $2d$ and $3d$, however, algorithms in higher dimensions have yet to be demonstrated, which rely on the efficient calculation of the power diagram (Laguerre diagram) in higher dimensions.
Here, we introduce an algorithm for computing power diagrams, which extends to any topological dimension.
We first evaluate the performance of the algorithm in $2d-6d$.
We then restrict our attention to four-dimensional settings, demonstrating that our power diagrams can be used to solve optimal quantization and semi-discrete optimal transport problems, whereby a prescribed mass of each power cell is achieved by computing an optimized power diagram.
}


\maketitle
\thispagestyle{empty}
\pagestyle{empty}

\section{Introduction}
The problem of transporting mass from one location to another so as to minimize total cost was studied by Gaspard Monge in 1781~\cite{Monge_1781}.
More generally, this problem consists of finding a map between two measures, and can be applied to problems in image processing~\cite{Balzer_2009,deGoes_2012,Ma_2018a,Ma_2018b,Meyron_2018,Meyron_2019}, machine learning~\cite{Peyre_2019} geometric processing~\cite{deGoes_2011}, rendering~\cite{Bonneel_2019,Paulin_2020}, resource allocation~\cite{Hartmann_2020} and particle-based simulations~\cite{Levy_2018,deGoes_2015}.
This problem also appears in the solution to partial differential equations, specifically when the two measures are continuous - see for example, the seminal work of Benamou and Brenier~\cite{Benamou_2000}.

Another (possibly more important) application is Villani's example of transporting bread from bakeries to caf\'es~\cite{Villani_2009}.
There are a discrete number of bakeries, each with a specific baking capacity, and there are a discrete number of caf\'es -- see Fig.~\ref{fig:bakeries}.
One may wish to distribute bread from bakeries to caf\'es so as to minimize the transport cost (Fig.~\ref{fig:cartoon-dot}).
This is a fully discrete optimal transport problem, since both the input and output measures are discrete.
This type of problem is important in shape- and image-matching problems, supply chain management and clustering~\cite{Solomon_2017}.
The revived interest in solving these types of optimal transport problems is due to the introduction of efficient algorithms that enable the computation of an optimal transport map for \textit{large} data sets.
These methods are mostly based on the idea of entropic regularization, a relaxation of the dual Kantorovich problem~\cite{Merigot_2017}.
A noteworthy algorithm for solving this discrete problem is known as the Sinkhorn-Knopp algorithm~\cite{Cuturi_2013}.

\begin{figure*}[h]
	\centering
	\begin{subfigure}[b]{0.475\textwidth}
		\resizebox{\textwidth}{!} {
			%
			%
			\includegraphics{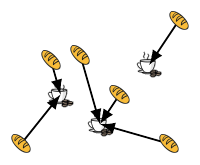}
		}
		\caption{Discrete optimal transport problem.
			Some caf\'es may be a in a busy part of town and require a lot of bread to be delivered.
		}
		\label{fig:cartoon-dot}
	\end{subfigure}
	\hspace{5pt}
	\begin{subfigure}[b]{0.475\textwidth}
		\resizebox{\textwidth}{!} {
			\includegraphics[scale=0.5]{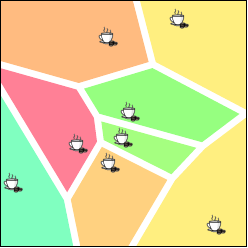}
		}
		\caption{Semi-discrete optimal transport problem. The city is divided such that each caf\'e serves an equal population (integral of the population density over region).}
		\label{fig:cartoon-sdot}
	\end{subfigure}
	\caption{Motivation for studying optimal transport: (a) transporting bread from bakeries to caf\'es and (b) assigning a population (with an assumed continuous density model) to caf\'es based on distance.}
	\label{fig:bakeries}
\end{figure*}

If, instead, we can construct a model for a continuous population density across the city, we may wish to partition the city so as to equidistribute the bread demand using the distance from the population to each caf\'e (Fig.~\ref{fig:cartoon-sdot}).
This would be a semi-discrete optimal transport problem.
If the cost function is the squared Euclidean distance, then we have an $L^2$ semi-discrete optimal transport problem.
The optimal transport problem then consists of finding this partition such that every caf\'e has the same bread demand.

The solution to the latter problem can be obtained with a power diagram (or Laguerre diagram).
The power diagram is a generalization of the Voronoi diagram in which sites are equipped with \textit{weights} that control the power distance from a point to a particular site.
The weights on these Voronoi sites can be tuned such that each region (a power cell) has an equal mass under the input density measure.
In 1987, Aurenhammer observed that a power diagram can be computed via the restriction of a Voronoi diagram in a higher dimensional space~\cite{Aurenhammer_1987,Aurenhammer_1991}.
That is, the power diagram is the intersection of a higher-dimensional Voronoi diagram with our domain (in the bakery example, the city) embedded to a higher dimensional space.
This concept of embedding has also been studied to achieve anisotropic Voronoi diagrams and meshes~\cite{Canas_2006_Surface_remeshing_arbitrary_codimension,Levy_2013_Vorpaline,Dassi_2015,Dassi_2016,Dassi_2017,Nivoliers_2015}.
In particular, L\'evy and Bonneel introduced the \textit{security radius theorem} which enabled the fast computation of the restricted Voronoi diagram~\cite{Levy_2013_Vorpaline}.
In 2015, this technique was used for semi-discrete optimal transport in three dimensions, and has since been used in fluid simulations \cite{Levy_2018,deGoes_2015} and astrophysics applications~\cite{Levy_2020}.
Power diagrams have also proven useful for computer graphics applications such as image stippling~\cite{Balzer_2009,deGoes_2012,Ma_2018a,Ma_2018b}, in which a blue noise sampling distribution is desirable to reduce aliasing effects that are produced by a regular sampling distribution, but produces better results than a white noise distribution.
Optimal transport has also recently been used to obtain blue noise sampling distributions for rendering applications~\cite{Bonneel_2019,Paulin_2020}.

Software implementations for computing power diagrams are available in \texttt{geogram} \cite{Levy_2016_Geogram}, \texttt{Voro++} \cite{Rycroft_2009} and \texttt{CGAL} \cite{CGAL_software,CGAL_Voronoi} in two and three dimensions.
These implementations often rely on clipping Voronoi polygons or polyhedra using algorithms such as Sutherland-Hodgman re-entrant clipping \cite{Sutherland_1974}, which are difficult to extend to a higher-dimensional setting.
A demonstration of the power diagram calculation in four or higher dimensions has yet to be demonstrated.
We make use of a simple result from polytope theory to convert between a facet-based and vertex-based representation of the power cells~\cite{Henk_2004_Basic_properties_of_convex_polytopes,Ziegler_1995_Lectures_on_Polytopes}.
It is well known that the number of vertices in a power cell grows exponentially with dimension, but the computation of the power diagram for semi-discrete optimal transport is nonetheless useful for higher-dimensional applications, particularly spatio-temporal simulations of partial differential equations, which would require four-dimensional power diagrams.
Furthermore, L\'evy and Bonneel's security radius theorem enables an efficient, emabarassingly parallel computation of the power cells, which makes the computation tractable in four dimensions.

The goal of this paper is to develop the theory and describe an algorithm for computing higher-dimensional power diagrams for semi-discrete optimal transport.
To our knowledge, this is the first demonstration of an algorithm for higher-dimensional power diagrams (in terms of topological dimension) in the literature.
We begin by reviewing the necessary background on semi-discrete optimal transport and power diagrams (Section~\ref{sec:background}) and then present our algorithm for computing power diagrams in a dimension-independent manner (Section~\ref{sec:algorithm}).
We evaluate the performance of the algorithm in $2d-6d$ (Section~\ref{sec:performance}) and then apply the algorithm to quantization and $L^2$ semi-discrete optimal transport in four dimensions (Section~\ref{sec:applications}), successfully demonstrating how a uniform target mass can be achieved under a prescribed density function.

\section{Background}
\label{sec:background}
Let us briefly review the optimal transport problem.
For a more complete review, we encourage the interested reader to see the works of L\'evy, \cite{Levy_2018,Levy_2015}, Peyr\'e~\cite{Peyre_2019}, Santambrogio~\cite{Santambrogio_2015} and Villani~\cite{Villani_2009,Villani_2003}.
The optimal transport problem consists of finding a map that transports a probability measure $\mu$ supported on a $d$-dimensional domain $\mathcal{X} \subset \mathbb{R}^d$ onto another probability measure $\nu$ supported on another domain $\mathcal{Y} \subset \mathbb{R}^d$.
The Monge formulation seeks the transport map $T : \mathcal{X} \to \mathcal{Y}$ from the following minimization statement:
\begin{align}\label{eq:monge-problem}
	\begin{split}
	T =& \arg\min_{T} \int\limits_{\mathcal{X}} c(x,T(x))\,\mathrm{d}\mu(\mathbf{x}),\\
	 &\mathrm{subject\ to:}\\ & \forall B \subset \mathcal{Y},\ \mu(T^{-1}(B)) = \nu(B),
 \end{split}
\end{align}
where $c(\cdot,\cdot) : \mathcal{X} \times \mathcal{Y} \to \mathbb{R}$ is the \emph{cost} of transporting mass from $\mathcal{X}$ to $\mathcal{Y}$.
The constraint on $T$ is a statement about conservation of mass for any Borel set $B \subset \mathcal{Y}$.

Kantorovich introduced a relaxation of the original Monge problem of Eq.~\ref{eq:monge-problem}, in which the \textit{transport map} is replaced by a \textit{transport plan}~\cite{Kantorovich_1958}, thereby reformulating Monge's problem as a linear programming problem with convex constraints.

\subsection{Semi-discrete optimal transport}
In the semi-discrete setting, the source $\mu$ is a continuous measure, and the target measure $\nu$ is discrete, i.e. it is a sum of $N$ Dirac masses: $\nu = \sum_{i=1}^N \nu_i \delta_{y_i}$.
These Dirac masses are each located at a site $\mathbf{y}_i \in \mathbf{Y}$ where $\mathbf{Y} \in \mathbb{R}^{d\times N}$ can be interpreted as a matrix in which each column contains the coordinates of a site $\mathbf{y}_i$.
The power cell $\powercell_i(\mathbf{Y},\mathbf{w})$ for the Dirac mass at $\mathbf{y}_i$ is defined as
\begin{align}\label{eq:power-cell}
	\powercell_i(\mathbf{Y},\mathbf{w}) = \left\{ \mathbf{x} \in \mathcal{X} \mid \forall j, c(\mathbf{x},\mathbf{y}_i) - w_i \le c(\mathbf{x},\mathbf{y}_j) - w_j \right\}.
\end{align}
The \emph{Power Diagram} (or Laguerre Diagram) is then the union of all power cells, which is a partition of $\mathcal{X}$.
We now restrict our attention to the case in which the transport cost is the squared Euclidean distance, so $c(\mathbf{x},\mathbf{y}_i) = \lvert|\mathbf{x}-\mathbf{y}_i\rvert|^2$.
The dual Kantorovich formulation can be stated as a maximization of the following energy functional:
\begin{align}\label{eq:energy-functional}
	\begin{split}
	E(\mathbf{Y},\mathbf{w}) &= \sum\limits_{i=1}^N \int\limits_{\mathrlap{\powercell_i(\mathbf{Y},\mathbf{w})}} \rho(\mathbf{x}) (\lvert|\mathbf{x}-\mathbf{y}_i\rvert|^2 - w_i)\,\mathrm{d}\mathbf{x} + \sum\limits_{i=1}^N\nu_i w_i.
	\end{split}
\end{align}
where we have replaced the mass $\mathrm{d}\mu(\mathbf{x}) = \rho(\mathbf{x})\mathrm{d}\mathbf{x}$ in terms of the prescribed density $\mathbf{\rho}(\mathbf{x})$.

\subsection{Optimizing the transport map}
In this paper, we will consider two settings: (1) we want to optimize the sites $\mathbf{Y}$ to minimize $E$ and (2) we want to optimize the weights $\mathbf{w}$ so as to maximize $E$.
The former problem is known as \emph{quantization}; a notable algorithm is \emph{Lloyd relaxation}~\cite{Lloyd_1982,Du_1999_CVT_Algorithm_Application} in which sites are iteratively moved to the centroids of their associated cells.
That is, at each iteration of Lloyd relaxation, each site $\mathbf{y}_i$ is updated to it's centroid $\mathbf{c}_i$:
\begin{align}\label{eq:lloyd-relaxation}
	\mathbf{c}_i = \frac{ \int\limits_{\mathrlap{V_i(\mathbf{Y})}} \rho(\mathbf{x}) \mathbf{x}\,\mathrm{d}\mathbf{x} } { \int\limits_{\mathrlap{V_i(\mathbf{Y})}} \rho(\mathbf{x})\,\mathrm{d}\mathbf{x} }.
\end{align}
Note that we have replaced $\powercell_i(\mathbf{Y},\mathbf{w}) = V_i(\mathbf{Y}) = \powercell_i(\mathbf{Y},\mathbf{0})$ - i.e. we recover the Voronoi diagram when the weights are all zero.
Lloyd relaxation is slow to converge, thus a gradient-based method is often employed, which requires the gradient of $E$ with respect to the sites $\mathbf{y}_i$:
\begin{align}\label{eq:de-dx}
	\frac{\mathrm{d}E}{\mathrm{d}\mathbf{y}_i} = 2 m_i (\mathbf{y}_i - \mathbf{c}_i),
\end{align}
where $m_i$ is the mass of the cell (the denominator in Eq.~\ref{eq:lloyd-relaxation}).
Some Newton-based methods have been proposed, which requires the computation of the Hessian $\mathrm{d}^2E/\mathrm{d}\mathbf{y}_i\mathrm{d}\mathbf{y}_j$~\cite{Levy_2018}.
For simplicity, we use a quasi-Newton approach, specifically the L-BFGS method~\cite{Liu_1989,Johnson_NLOPT}.

In the second setting, we wish to maximize Eq.~\ref{eq:energy-functional} by optimizing the $N$ weights $\mathbf{w}$.
We can treat this as a minimization of $-E(\mathbf{Y},\mathbf{w})$, and use the derivatives:
\begin{align}\label{eq:de-dw}
	\frac{\mathrm{d}E}{\mathrm{d}w_i} = \nu_i - \int\limits_{\mathrlap{\powercell_i(\mathbf{Y},\mathbf{w})}} \rho(\mathbf{x})\,\mathrm{d}\mathbf{x} = \nu_i -  m_i,
\end{align}
where $\nu_i$ is the target (prescribed) mass of the power cell $\powercell_i(\mathbf{Y},\mathbf{w})$.
Since the energy is concave, it admits a unique maximizer~\cite{Levy_2018}.
Previous works have analyzed Newton-based approaches~\cite{Kitagawa_2016} for optimizing the weights, as well as multi-scale methods for large scale transport problems~\cite{Merigot_2011}.
These approaches have been restricted to two- and three-dimensions due to existing software implementations for computing power diagrams~\cite{Levy_2016_Geogram,Rycroft_2009,CGAL_software,CGAL_Voronoi}.
Here we strive to compute power diagrams in any dimension to support the solution of higher-dimensional semi-discrete optimal transport problems.

\section{Power diagrams in higher dimensions}
\label{sec:algorithm}
Given $N$ input sites $\mathbf{Y} = \{ \mathbf{y}_i \mid \mathbf{y}_i \in \mathbb{R}^d \}_{i = 1}^N$ and some weights $\mathbf{w} \in \mathbb{R}^N$ (a scalar weight is associated with each site), our goal is to compute the power diagram (with cells defined by Eq.~\ref{eq:power-cell}) \emph{restricted} to the domain $\mathcal{X} \subseteq \mathbb{R}^d$.
This domain can either be specified as a mesh of $d$-polytopes or $d$-simplices, therefore, the term ``restricted" means that we intersect $\mathcal{X}$ with $\powercell_i(\mathbf{Y},\mathbf{w})$ for each power cell.

Existing software implementations have been restricted to two- and three-dimensions and use geometric algorithms such as Sutherland-Hodgman re-entrant clipping~\cite{Sutherland_1974} to compute the intersection of a polygon or polyhedron with a Voronoi bisector.
The latter is needed in order to compute the vertex coordinates of a clipped polyhedron so as to apply the security radius theorem~\cite{Levy_2013_Vorpaline} and determine if clipping should be terminated.
Here, we introduce an alternative view of the clipping procedure, which extends to any dimensional domain $\mathcal{X}$.
Before describing our algorithm, we need a ``simple" result from polytope theory.

\subsection{Voronoi polytopes are simple}

Our algorithm relies on the fact that Voronoi polytopes are \textit{simple}, since they are the duals of Delaunay simplices.
This enables the conversion between a facet-based and vertex-based representation of a Voronoi polytope.
A convex $d$-dimensional polytope $\powercell$ can be described as the intersection of $m$ halfspaces.
Thus each point in the polytope can be described as the bounded solution set of $m$ linear inequalities~\cite{Henk_2004_Basic_properties_of_convex_polytopes}
\begin{align}
	\mathcal{H}(\powercell) = \{ \mathbf{x} \in \mathbb{R}^d \mid \mathbf{A}^T(\mathbf{x} - \mathbf{x}_0) \le \mathbf{0} \},
\end{align}
where each column of $\mathbf{A} \in \mathbb{R}^{d\times m}$ represents the normal to a facet.
This is known as the \textit{H-Representation} or \textit{HRep} for short, denoted by $\mathcal{H}$.
The \textit{vertex enumeration problem} consists of converting the HRep to the \textit{V-Representation} (\textit{VRep}, denoted by $\mathcal{V}$), which is a description of the polytope as the convex hull of its $n$ vertices:
\begin{align}
	\mathcal{V}(\powercell) = \mathrm{conv}(\mathbf{V}),
\end{align}
where $\mathbf{V} \in \mathbb{R}^{d \times n}$ is a matrix with each vertex $\mathbf{v}_i \in \mathbb{R}^d$ stacked columnwise.

In the following, assume we have information regarding which facets are incident to every vertex, which is stored in the \textit{vertex-facet-incidence matrix}, and will be denoted as $\mathbf{F}(\cdot) : \mathbb{N} \to \mathbb{Z}^k$, for some $k \ge 0$.
Observe that a vertex $v$ is identified using a nonnegative integer, whereas a facet $b$ is labeled as an integer.
This enables the distinction between a Voronoi bisector ($b \ge 0$) between two sites, and a mesh facet ($b < 0$), which is a $(d-1)$-dimensional facet of the input mesh.
In order to compute the intersection of a $d$-polytope with a halfspace, we need the following definition from Henk~\cite{Henk_2004_Basic_properties_of_convex_polytopes}.

\begin{definition}[\textit{simple polytope}]
  A $d$-polytope $\powercell$ is said to be \emph{simple} if every vertex is incident to exactly $d$ facets.
  A property of simple polytopes is that their dual polytopes are simplices.
\end{definition}

Thus, for Voronoi polytopes, $\mathbf{F}(\cdot) : \mathbb{N} \to \mathbb{Z}^d$ (i.e. $k = d$).
We then have the following result, which is critical for our intersection algorithm.
\begin{corollary}
The edges $\mathcal{E}(\powercell)$ can be identified from the following relation on $\mathcal{V}(\powercell)$:
\begin{equation}\label{eq:polytope-edges}
  \mathcal{E}(\powercell) = \left\{ e = (v_0,v_1) ~\big|~\lvert \mathbf{F}(v_0) \cap \mathbf{F}(v_1)\rvert = d-1 \right\}.
\end{equation}
\end{corollary}
The result is certainly true for $d \le 4$~\cite{Henk_2004_Basic_properties_of_convex_polytopes}, but is important to consider for higher-dimensions.
For \textit{simple} polytopes, the dual of $\powercell$ is a simplex (the Delaunay simplex) from which the entire set of facets is trivially constructed.
Furthermore, the hierarchy of the facets of a polytope can be obtained from the corresponding facet hierarchy of its dual~\cite{Henk_2004_Basic_properties_of_convex_polytopes}, which contains the edges.
Therefore, the edges $\mathcal{E}(\powercell)$ of a simple $d$-polytope $\powercell$ can be derived purely from the vertex-facet-incidence matrix of $\powercell$ by traversing the facet-hierarchy of the dual simplex in reverse order.
Thus we can determine if there is an edge between two vertices in the VRep of $\powercell$ if they share $d-1$ common facets.
\begin{figure}
	\centering
	\resizebox{0.45\textwidth}{!} {
		\includegraphics{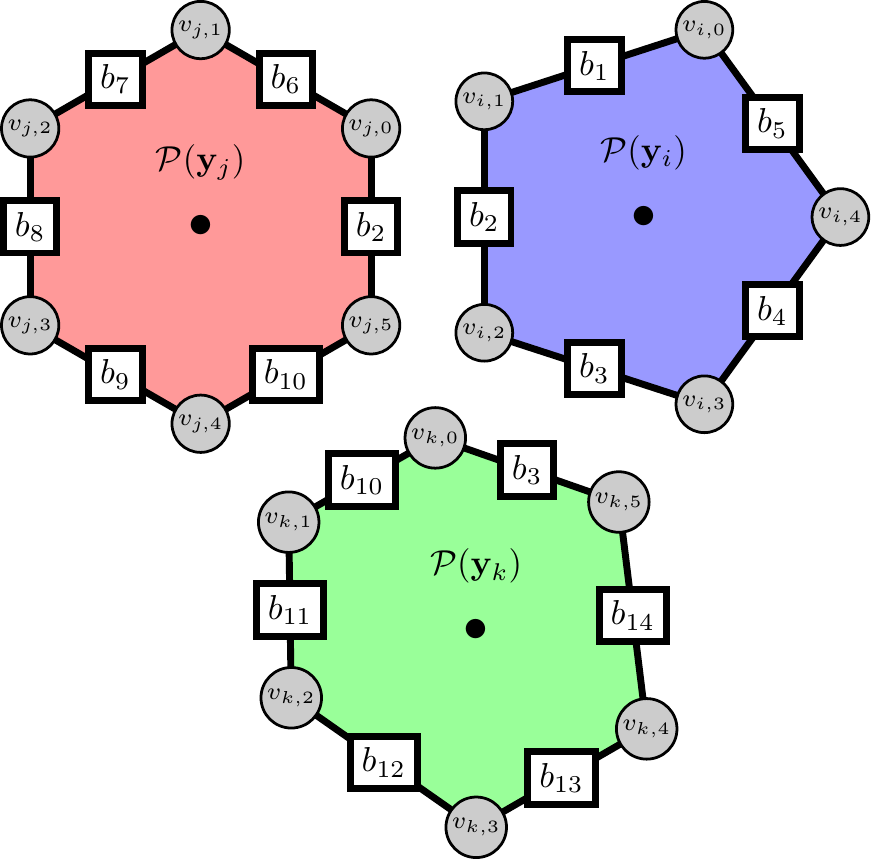}
		}
	\caption{Computing the edges of a polytope from the vertex-facet-incidence relations.
	Three Voronoi cells are shown in the different colors and bisectors (possibly shared) are labeled in the squares.
	}
	\label{fig:incidence-example}
\end{figure}

For example, in Fig.~\ref{fig:incidence-example}, $\mathbf{F}(v_{i,1}) = \{b_1,b_2\}$ and $\mathbf{F}(v_{i,2}) = \{b_2,b_3\}$.
Since vertices $v_{i,1}$ and $v_{i,2}$ share a single common facet ($b_2$), then they also share an edge.
However, vertices $v_{j,3}$ and $v_{j,1}$ do not share an edge because $\mathbf{F}(v_{j,3}) \cap \mathbf{F}(v_{j,1}) = \emptyset$.
Furthermore, we can identify that vertices $v_{i,2}$, $v_{j,5}$ and $v_{k,0}$ are all symbolically equivalent because they share a common bisector, thus enabling these vertices to be merged as a single vertex if necessary.
In fact, these vertices represent the Delaunay simplex between sites $\mathbf{y}_i$, $\mathbf{y}_j$ and $\mathbf{y}_k$.

Each Voronoi cell will be computed from the intersection of a finite number of halfspaces, thus it is important to ensure that each cell remains a simple polytope after intersecting it with a halfspace, so that we can iteratively apply Eq.~\ref{eq:polytope-edges} for each bisector.
\begin{proposition}
  A simple $d$-polytope $\powercell \subset \mathbb{R}^d$ intersected with a halfspace $\mathcal{H}^+$ produces a simple polytope.
\end{proposition}
\begin{proof}
  Assuming the intersection of $\powercell$ with $\mathcal{H}^+$ is non-empty, it suffices to show that every vertex of the new polytope will be incident to exactly $d$ facets.
  One way of describing the halfspace $\mathcal{H}^+$ is by using a unique point $\mathbf{x}_0$ and normal to the dividing hyperplane $\mathbf{n} \in \mathbb{R}^d$.
  Since $\powercell$ is simple, we can determine its edges $\mathcal{E}(\powercell)$.
  Every vertex of the original polytope $v \in \mathcal{V}(\powercell)$ can then be classified as to whether it is in the halfspace as
  \begin{equation}\label{eq:vertex-classify}
    \mathcal{V}^+(\powercell) = \left\{ v \in \mathcal{V}(\powercell)~\big|~(\mathbf{x}(v) - \mathbf{x}_0)\cdot \mathbf{n} > 0 \right\},
  \end{equation}
	where $\mathbf{x}(v) \in \mathbb{R}^d$ are the coordinates of vertex $v$.
  The vertices created from the intersection are then computed by intersecting each edge with the dividing hyperplane (note we can filter which edges are intersected by finding edges with one vertex in $\mathcal{V}^+(\powercell)$ and one that is not).
  Denote the set of intersection vertices as $\mathcal{S}$, therefore, the new vertices of the polytope $Q$ are $\mathcal{V}(Q) = \mathcal{V}^+(\powercell) \cup \mathcal{S}$.
  There are two cases to consider.
  First, the vertices $\mathcal{V}^+(\powercell)$ are clearly adjacent to $d$ facets since they were not affected by the intersection.
  Second, the vertices in $\mathcal{S}$ were each created from the intersection of an edge with the hyperplane.
  Since edges are adjacent to $d-1$ facets (Eq.~\ref{eq:polytope-edges}) and the dividing hyperplane, itself, defines a facet of $Q$, then each intersection vertex is adjacent to $d$ facets.
\end{proof}
We are now ready to describe our algorithm for computing power cells.
\subsection{Computing the power cells}
\label{sec:power-cells}
In fact, we simply compute the Voronoi cells!
As noted earlier, Aurenhammer observed that the power diagram is the intersection of a higher-dimensional Voronoi diagram with our $d$-dimensional space lifted to the same ambient dimension of this Voronoi diagram~\cite{Aurenhammer_1987,Aurenhammer_1991}.
We use the result of L\'evy~\cite{Levy_2015} to lift the Voronoi sites (equipped with weights) to a $(d+1)$-dimensional space, thereby obtaining $\mathbf{Z} \in \mathbb{R}^{(d+1)\times N}$ in which
\begin{equation}\label{eq:lift-sites}
	\mathbf{z}_i = \left(\mathbf{y}_i^T,\ \sqrt{\max(\mathbf{w}) - w_i}\right)^T, \quad \forall i = 1,2,\dots,N.
\end{equation}
We additionally need to lift our domain of interest to $\mathbb{R}^{d+1}$, which is done by simply appending a zero to the coordinates of the domain (the vertices of the mesh).
As a result, we only need to compute the Voronoi diagram in $\mathbb{R}^{d+1}$.
That is, the power cells are $\powercell_i(\mathbf{Y},\mathbf{w}) = V_i(\mathbf{Z})$.

Although the domains we study in this paper can be entirely described with a single polytope (a $d$-cube), we will consider the general case in which the domain to be clipped against is represented as a simplicial mesh.
The mechanics of the algorithm are the same, but this description lends well to future work for the interested reader.

We begin with a single element, here represented by a $d$-simplex $\kappa$.
The vertex-facet incidence relations for each vertex of $\kappa$ are the \textit{mesh facets} (the $(d-1)$-simplices) incident to a particular vertex.
As mentioned earlier, these mesh facets are labeled with negative integers so as to distinguish them from Voronoi bisectors.
Our goal is to compute the Voronoi cell associated with site $\mathbf{z}_i$, clipped with the element $\kappa$.
In the following description, please follow along with Fig.~\ref{fig:voronoi-clipping}.

\begin{figure*}[t]
	\centering
	\begin{subfigure}[b]{0.245\textwidth}
	\resizebox{0.99\textwidth}{!} {
	  \includegraphics{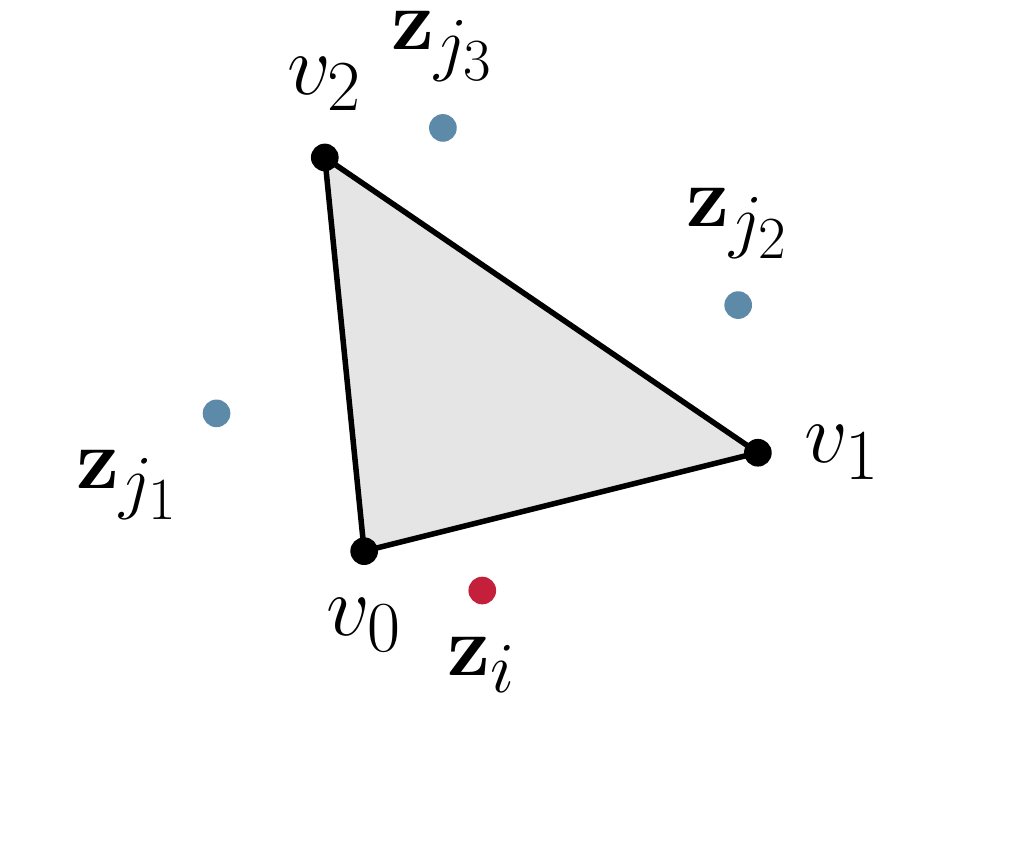}
	}%
	\caption{Initial simplex $\kappa$.}
	\end{subfigure}
	\begin{subfigure}[b]{0.245\textwidth}
	\resizebox{0.99\textwidth}{!} {
		\includegraphics{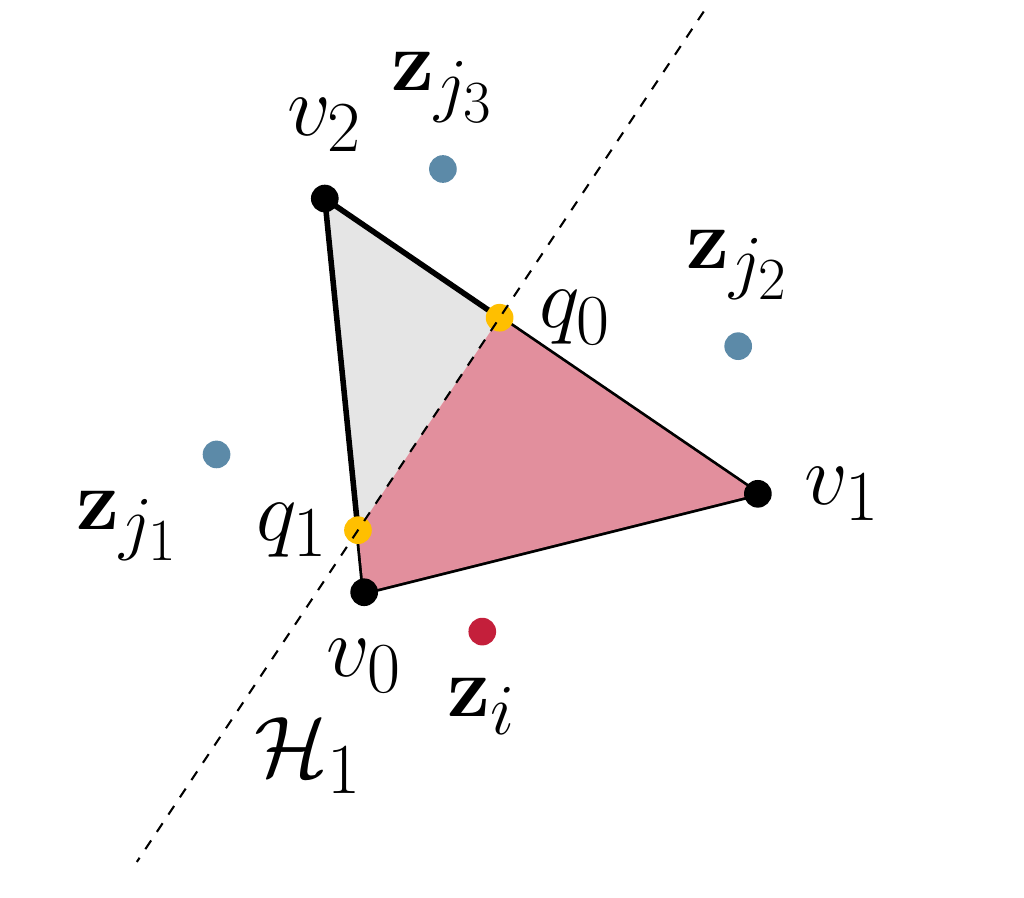}
	}%
	\caption{Clipping with $\mathcal{H}_1$.}
	\end{subfigure}
	\begin{subfigure}[b]{0.245\textwidth}
	\resizebox{0.99\textwidth}{!} {
		\includegraphics{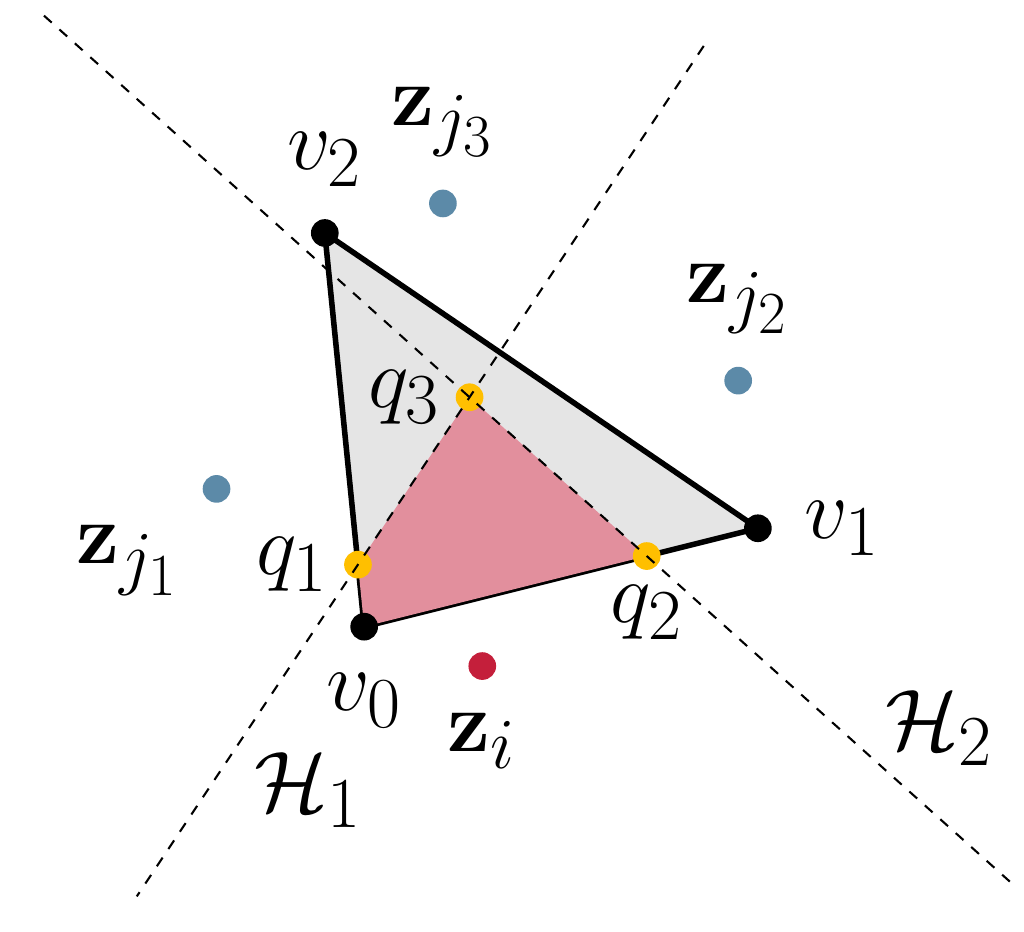}
	}%
	\caption{Clipping with $\mathcal{H}_2$.}
	\end{subfigure}
	\begin{subfigure}[b]{0.245\textwidth}
	\resizebox{0.99\textwidth}{!} {
		\includegraphics{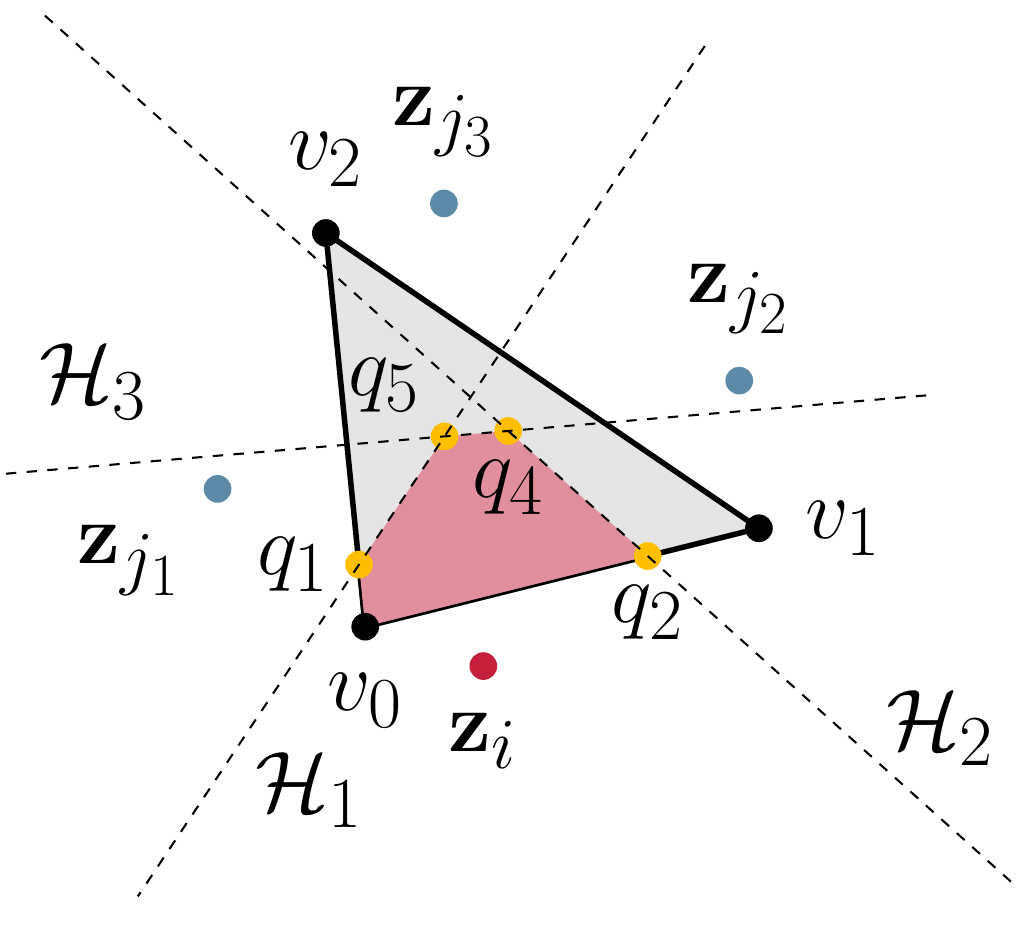}
	}%
	\caption{Clipping with $\mathcal{H}_3$.}
	\end{subfigure}
	\caption{Computing the intersection of a mesh element (here, a simplex) with a Voronoi cell defined by the site $\mathbf{z}_i$.
	Clipping starts with the Voronoi bisector $\mathcal{H}_1$ defined by the nearest neighbor to site $\mathbf{z}_i$ ($\mathbf{z}_{j_1}$) and proceeds to the next neighbors, thus clipping by $\mathcal{H}_2$ and then $\mathcal{H}_3$ from left to right.
	}
	\label{fig:voronoi-clipping}
\end{figure*}

We begin by clipping $\kappa$ with the Voronoi bisector $\mathcal{H}_1$ defined by $\mathbf{z}_i$ and it's nearest neighbor, $\mathbf{z}_{j_1}$.
In order to determine the clipped polytope (in red), we first identify the edges of $\kappa$ using Eq.~\ref{eq:polytope-edges}.
Next, we determine which edges have vertices on either side of the Voronoi bisector $\mathcal{H}_1$.
Those with vertices on either side of $\mathcal{H}_1$ are then intersected with $\mathcal{H}_1$, creating new vertices $q_0$ and $q_1$.
We then set $\mathbf{F}(q_0) = \mathbf{F}(v_1) \cap \mathbf{F}(v_2) \cup \{b_{i,j_1}\}$, where $b_{i,j_1}$ is the unique integer label assigned to the Voronoi bisector between sites $i$ and $j_1$.
Furthermore, $\mathbf{F}(q_1) = \mathbf{F}(v_0) \cap \mathbf{F}(v_2) \cup \{ b_{i,j_1} \}$.
The geometric coordinates of $\mathbf{q}_0$ and $\mathbf{q}_1$ are computed and used to check if clipping should be terminated, using the radius of security theorem~\cite{Levy_2013_Vorpaline}.

If clipping should continue, we then proceed to the next nearest neighbor of site $\mathbf{z}_i$, which is $\mathbf{z}_{j_2}$.
A similar procedure ensues, this time, starting with the polytope defined by vertices $\{v_1,q_0,q_1,v_2\}$.
Extracting the edges (Eq.~\ref{eq:polytope-edges}) and intersecting those with vertices on either side of $\mathcal{H}_2$ reveals that $q_0$ is no longer in the Voronoi cell.
New vertices $q_2$ and $q_3$ are introduced from the intersection of the edges with $\mathcal{H}_2$.
Then $\mathbf{F}(q_2) = \mathbf{F}(v_0) \cap \mathbf{F}(v_1) \cup \{b_{i,j_2}\}$ and $\mathbf{F}(q_3) = \mathbf{F}(q_0) \cap \mathbf{F}(q_1) \cup \{b_{i,j_2}\}$.

Proceeding in a similar fashion to the next nearest neighbor $\mathbf{z}_{j_3}$ clips the current polytope defined by vertices $\{v_0,q_2,q_3,q_1\}$ (no need for them to be ordered counterclockwise).
We then end up with the polytope $\{v_0,q_2,q_4,q_5,q_1\}$ where $\mathbf{F}(q_4) = \mathbf{F}(q_2) \cap \mathbf{F}(q_3) \cup \{b_{i,j_3}\}$ and $\mathbf{F}(q_5) = \mathbf{F}(q_1) \cap \mathbf{F}(q_3) \cup \{b_{i,j_3}\}$.
We suppose the clipping terminates after checking the coordinates of this polytope with the radius of security theorem.

This procedure is described in Algorithm~\ref{alg:polytope-clipping}.
The inputs to the algorithm are a Voronoi site $\mathbf{z}_i$, a mesh element $\kappa$ and a nearest neighbor structure $\mathcal{N}$ that computes the nearest sites upon request.
This structure is initialized to $50$ nearest neighbors at the start of clipping - we expect well distributed mesh vertices to have a simplex valency of $6$ in $2d$, $15-20$ in $3d$ and about $120$ in $4d$~\cite{Caplan_2019_PhD}.
Whenever the bounds of the nearest neighbors are reached, an additional $10$ neighbors are appended to the current neighbors.
As we will examine in the results section, the cost of computing the nearest neighbors is overshadowed by clipping, simplex decomposition and numerical integration.

Clipping continues until the security radius is reached, which is equal to twice the maximum distance from any vertex in $\powercell$ to $\mathbf{z}_i$ (we absorb the factor of 2 into the definition of the radius of security - a slight modification of L\'evy's definition).
Each edge is then intersected with the bisector defined by sites $\mathbf{z}_i$ and $\mathbf{z}_j$ - the $j^{\mathrm{th}}$ nearest neighbor to $\mathbf{z}_i$.
When an intersection occurs, only one of $e_0$ or $e_1$ (the edge vertices) is in the halfspace defined by $\mathcal{H}$ which includes $\mathbf{z}_i$ (i.e. $\mathcal{H}^+$).
This vertex, along with the intersection vertex is then appended to the set of vertices defining the clipped polytope $Q$, and the vertex-facet incidence relations are updated.

In the current work, we parallelize the clipping procedure over the Voronoi sites in contrast to a parallelization over the mesh elements.
The reason is because we keep our domains are defined by a single $d$-cube but we seek power diagrams defined by sometimes millions of Voronoi sites.
Furthermore, we parallelize the computation on the CPU, specifically with \texttt{OpenMP}, however the computation is well suited for parallelization on the GPU, which will be explored as future work.

\begin{algsimple}[h!]
	\begin{algcode}[computeVoronoiCell]
		\zi \Inputs site $\mathbf{z}_i$, element $\kappa$, neighbors $\mathcal{N}$
		\zi \Outputs Voronoi cell as a $d$-polytope $\powercell$
		\zi
    \li $\powercell \gets \kappa$\,\,\Comment{initialize to domain element}
    \li $j = 1$\,\,\Comment{start with nearest neighbor}
		\li $\mathbf{z}_j = \mathcal{N}(\mathbf{z}_i,j)$
    \li \While $\lvert|\mathbf{z}_i-\mathbf{z}_j\rvert| < \texttt{securityRadius}(\mathbf{z}_i,\powercell)$ \addindent
      \li $\mathcal{H} \gets \mathcal{H}(\mathbf{z}_i,\mathbf{z}_{j})$\,\,\Comment{define bisector}
      \li $Q \gets \emptyset$\,\,\Comment{initial clipped polytope}
      \li $\mathcal{E} \gets \mathcal{E}(\powercell)$\,\,\Comment{from Eq.~\ref{eq:polytope-edges}}
      \li \For $e = (e_0,e_1) \in \mathcal{E}$
        \li $s_0 = \texttt{side}(e_0,\mathcal{H})$\Label{algline:side0}
        \li $s_1 = \texttt{side}(e_1,\mathcal{H})$\Label{algline:side1}
        \li \If $s_0 \equiv s_1$\addindent\,\,\Comment{no intersection}
        \li \Continue\remindent
        \li $p \gets e_0$ \Or $e_1$ (whichever has $s > 0$)
        \li $\mathbf{q} \gets e \cap \mathcal{H}$ and set $q \gets \lvert Q\rvert$
        \li $\mathbf{F}(q) \gets \mathbf{F}(e_0) \cap \mathbf{F}(e_1) \cup b(\mathcal{H})$\label{algline:voronoi-facet-label}
        \li $Q \gets Q \cup \{p ,q\}$ \Label{algline:voronoi-append-vertices}
      \End
      \li $\powercell \gets Q$\,\,\Comment{update current polytope}
      \li $j = j +1$ \Comment{proceed to next neighbor}
			\li $\mathbf{z}_j = \mathcal{N}(\mathbf{z}_i,j)$
      \remindent
    \End
	\end{algcode}
	\caption{Computing the intersection of a mesh element $\kappa$ with a Voronoi cell defined by a site $\mathbf{z}_i$.}
	\label{alg:polytope-clipping}
\end{algsimple}

Eq.~\ref{eq:polytope-edges} suggests that extracting the polytope edges is quadratic in the number of vertices in the current clipped polytope.
Since the number of vertices in a Voronoi polytope is exponential in the dimension of the polytope, this could incur a significant computational cost - especially since this extraction is repeatedly performed until the radius of security is reached.
Therefore, we reduce the computational cost by (1) retaining any edges that lie entirely within $\mathcal{H}^+$ (the side of the hyperplane containing the site $\mathbf{z}_i$), (2) updating the endpoints of any intersected edges with indices of the new intersection vertices and (3) computing new edges that lie exactly on $\mathcal{H}$ using Eq.~\ref{eq:polytope-edges}.
Step (3) is still quadratic in the number of vertices that lie on $\mathcal{H}$, but since these define a $(d-1)$-dimensional polytope, the cost remains reasonable for our applications ($d \le 6$).

\paragraph{A note regarding exactness}
In this paper, we are purely interested in the \textit{geometry} of the power cells and do not extract any \textit{topological} information corresponding to the dual Delaunay mesh.
Although this topological information could be extracted by merging vertices that have identical symbolic information (e.g. vertices $v_{i,2}$, $v_{j,5}$ and $v_{k,0}$ in Fig.~\ref{fig:incidence-example}), this computation is susceptible to numerical precision issues.
Exact geometric predicates, such as an \texttt{orientnd} function, similar to Shewchuk's \texttt{orient2d} and \texttt{orient3d} functions could be used to detect cosphericities~\cite{Shewchuk_1996_Adaptive_Precision}.
In the context of our algorithm, the calls to the \texttt{side} function on Lines~9 and 10 should be replaced with a \texttt{side\_rd} function~\cite{Levy_2016_PCK} which computes the intersection of a $r$-dimensional simplex with the bisector $\mathcal{H}$.
In order to determine which $r$-simplex should be used, we simply need to track which simplex facets of the original mesh are incident to each vertex.
This can be done by tracking the ``simplex vertices" of every clipped vertex: $\mathbf{S}(v)$.
Any time an intersection point is introduced, the new vertex inherits all the simplex vertices of the endpoint edge vertices.
For example, in the rightmost diagram of Fig.~\ref{fig:voronoi-clipping}, $\mathbf{S}(q_5) = \mathbf{S}(q_3) \cup \mathbf{S}(q_1) = \mathbf{S}(q_0) \cup \mathbf{S}(q_1) = \{v_0,v_1,v_2\}$, meaning $q_5$ is the intersection of the $d$-dimensional simplex $\kappa = (v_0,v_1,v_2)$ with the bisectors $\mathcal{H}_1 \cap \mathcal{H}_2 \cap \mathcal{H}_3$.

\subsection{Numerical integration}
Computing the objective function and gradients of Eqs.~\ref{eq:energy-functional},~\ref{eq:de-dx} and \ref{eq:de-dw} requires the calculation of integrals over the resulting Voronoi polytopes.
To perform the integration, we decompose the $d$-polytopes into a set of $d$-simplices by introducing a vertex at the centroid of the polytope and then recurse through the $(d-1)$-facets, continuously introducing vertices at these facet centroids until the edges are reached (at which point the triangulation is trivial).
The simplicial decomposition of the $d$-polytope is obtained by reconnecting the simplicial facets with the centroids.
This mesh of $d$-simplices is then used to perform the integration with numerical quadrature rules of Stroud~\cite{Stroud_1971_Approximate_calculation_of_multiple_integrals}.
Although this evaluation is also performed in parallel over the integration simplices, the integral evaluation is a bottleneck in our algorithm, which will be discussed in the results section.

\subsection{Visualizing the power diagram}

In the results section, we will present visualizations of four-dimensional power diagrams, thus it is necessary to briefly describe our procedure for doing so.
Some previous work in $4d$ mesh visualization includes the work of Caplan~\cite{Caplan_2019_PhD} and Belda-Ferr\'in~\cite{BeldaFerrin_2020}.

Our method for visualizing a $d$-dimensional mesh first consists of identifying which polytopes of this mesh are cut by some input viewing volume, represented as a $(d-1)$-dimensional hyperplane with point $\mathbf{x}_0$ and normal $\mathbf{n}$.
Next, each polytope $\powercell$ must be intersected with the hyperplane.
We cull polytopes that are not clipped if all their vertices lie on the same side of the hyperplane, using the predicate in Eq.~\ref{eq:vertex-classify}.
Otherwise, the polytope is clipped by the hyperplane.
We then identify which edges have vertices on either side of the hyperplane, and then compute the intersection point of the hyperplane with the edge.
Each intersection point is then appended to a list of vertices that lie \textit{exactly} on the hyperplane, which defines a $(d-1)$-dimensional polytope that can be directly visualized if a suitable visualization dimension (such as $2d$ or $3d$) is reached.

For high-dimensional ($d > 4$) applications, a recursive procedure can be applied to this $(d-1)$-dimensional mesh along with a $(d-2)$-dimensional hyperplane.
In the results section, we only visualize $4d$ meshes through (1) the intersection of a $3d$ viewing volume or (2) a second intersection with a $2d$ clipping plane.
In order to view the polyhedra resulting from the intersection, we tetrahedralize the convex polyhedra using a Delaunay triangulator~\cite{Si_2005_TetGen} and visualize the mesh edges by applying Eq.~\ref{eq:polytope-edges} to extract the edges of each polyhedron clipped to the $3d$ viewing volume.

\section{Performance}\label{sec:performance}
In this section we evaluate the performance of our algorithm for computing $d$-dimensional power diagrams (of $d$-polytopes) for $d = 2$ to $d = 6$.
All tests are performed with an Intel Xeon W-2145 CPU at 3.70GHz with 8 cores (16 threads).

Our first performance test consists of distributing $N$ points within a unit $d$-cube and measuring the time to compute the power diagram.
We consider two point distributions: (1) a completely random white noise distribution and (2) a blue noise distribution.
The reason we consider the former is because the nearest neighbor calculation is more costly for a white noise distribution of points, whereas the number of nearest neighbors is more uniform in the case of a blue noise distribution.
The blue noise sampling distributions were obtained with the \texttt{SpokeDarts} software of Mitchell et al.~\cite{Mitchell_2018}.
An example of the resulting Voronoi diagrams for four-dimensional white and blue noise point distributions, sliced along the fourth dimension, are shown in Fig.~\ref{fig:point-distributions}.

\begin{figure*}[h!]
	\centering
	\begin{subfigure}{0.475\textwidth}
		\includegraphics[width=0.8\textwidth]{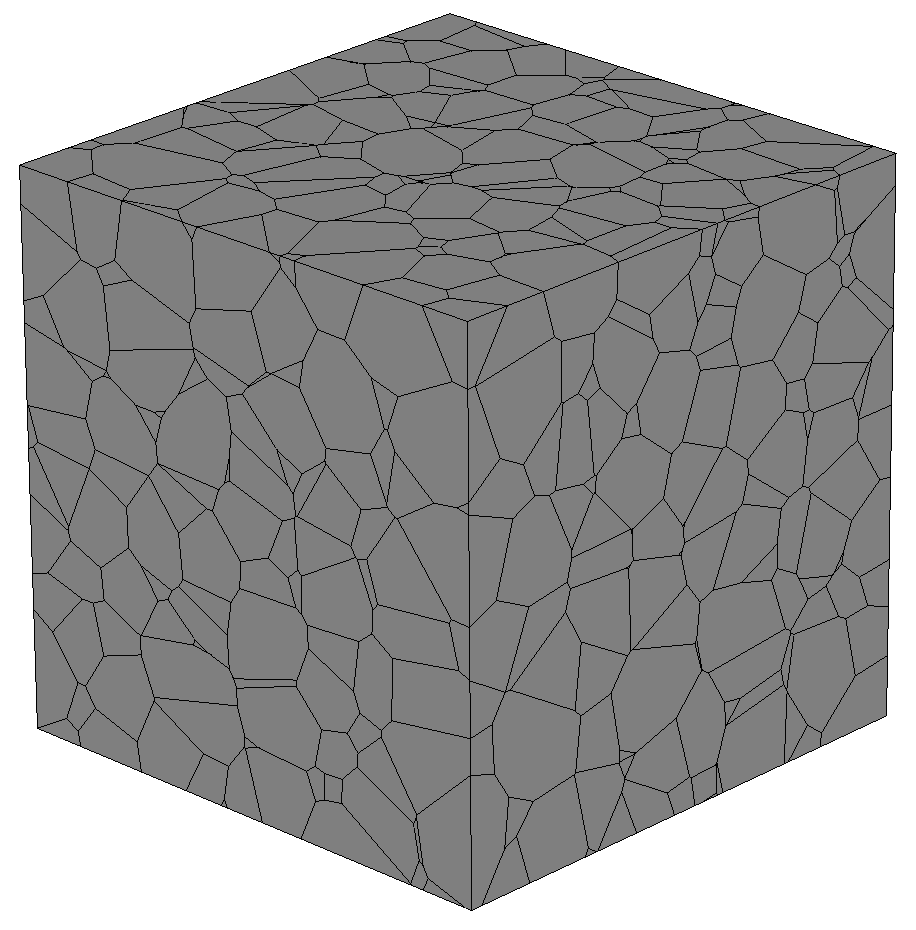}
		\caption{White noise.}
	\end{subfigure}
	\begin{subfigure}{0.475\textwidth}
			\includegraphics[width=0.8\textwidth]{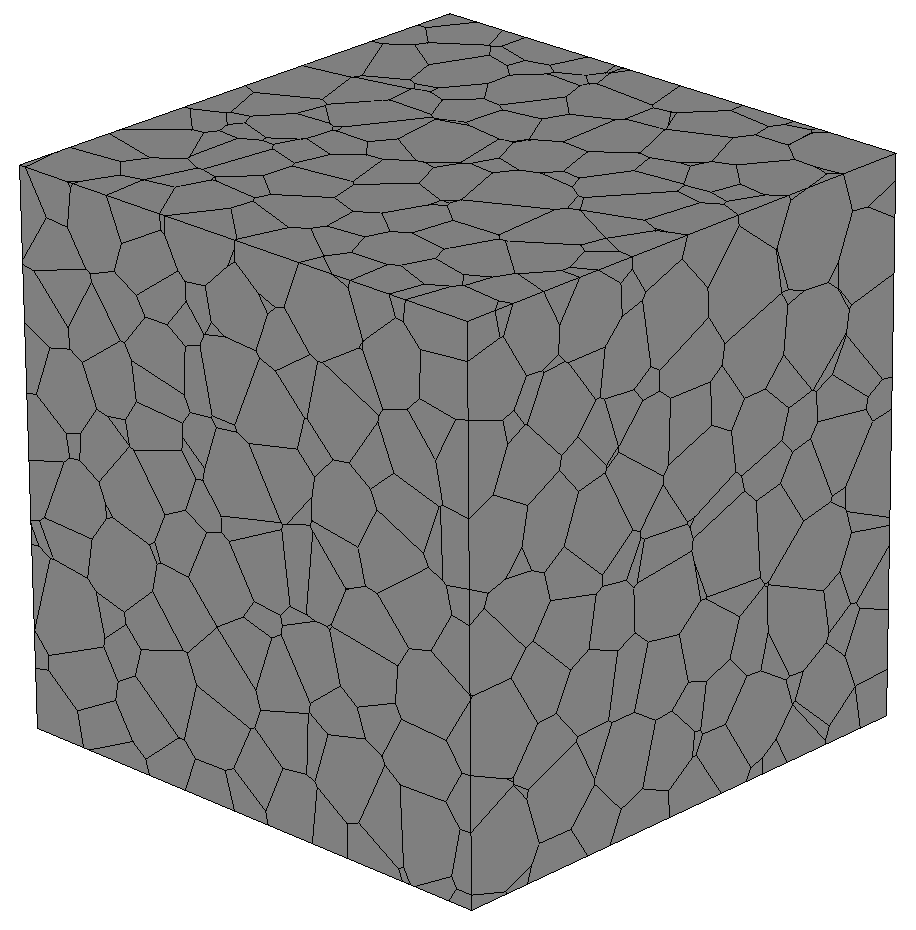}
			\caption{Blue noise.}
	\end{subfigure}
  \caption{Four-dimensional Voronoi diagrams for white and blue noise point distributions for $16,332$ sites embedded in $4d$.
	The images represent the slices of the Voronoi diagram along a hyperplane at $t = 0$ where $t$ is the fourth dimension.
	}
	\label{fig:point-distributions}
\end{figure*}

Performing the calculation in $2$- through $6$-dimensional domains (where the ambient dimension is equal to the topological dimension of the domain) reveals a good scaling of the algorithm in low dimensions, particular in $2d$ through $4d$.
Fig.~\ref{fig:cost2d-6d} reveals that our calculation is more costly for random point distributions than for blue noise distributions.
Also observe that the cost of computing the Voronoi diagram in $5d$ and $6d$ is significantly higher than in low dimensions, which is due to the fact that the number of vertices in a Voronoi polytope grows exponentially with the dimension.
Fig.~\ref{fig:vertices-facets-2d-6d} shows how the number of vertices and facets grows with the input number of points (Voronoi sites) for a white noise distribution.

\begin{figure}[h!]
	\centering
	\resizebox{0.4\textwidth}{!}{
	\begin{tikzpicture}
		\node at (0,0) {\includegraphics[width=0.5\textwidth]{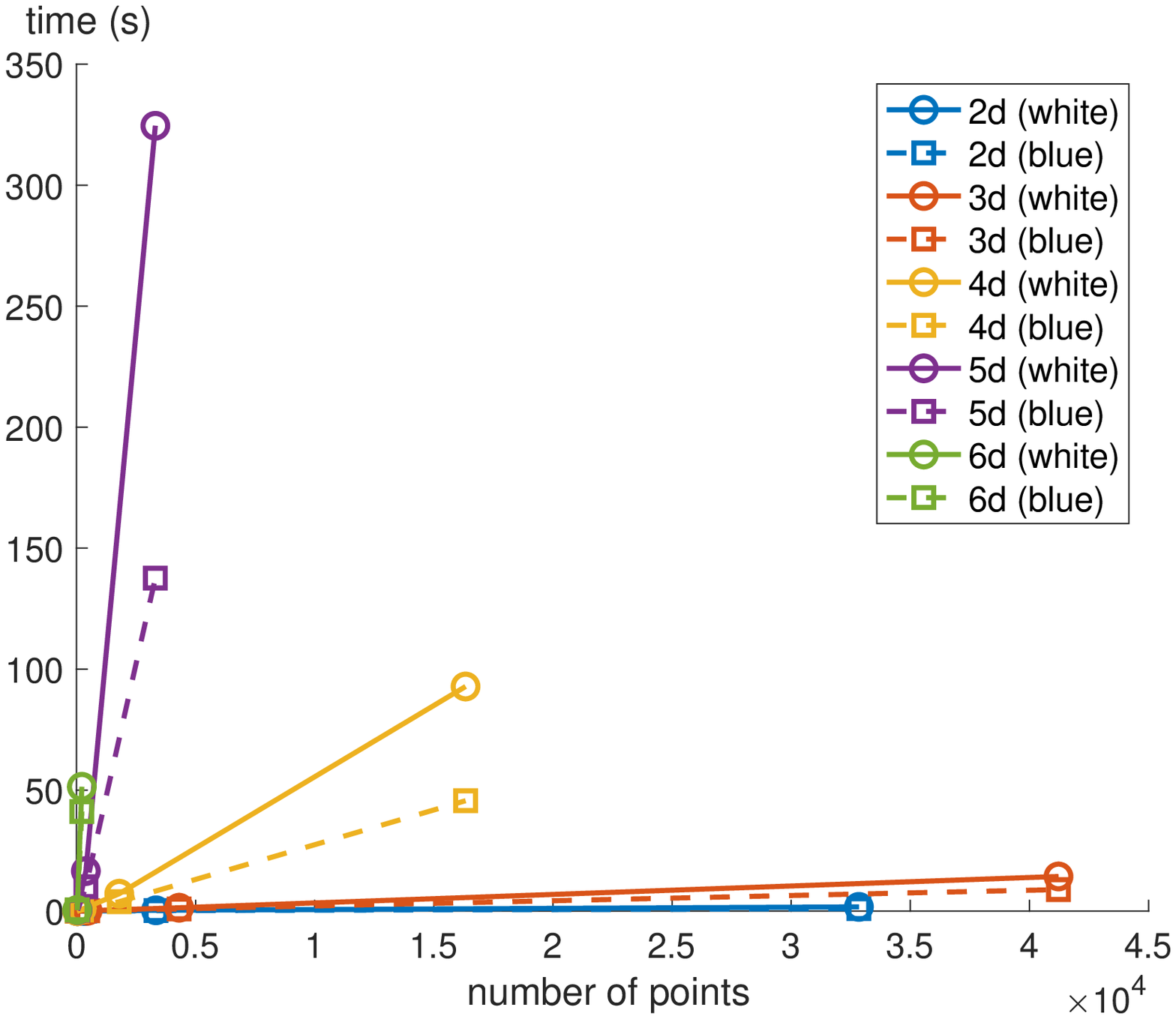}};
		\node at (1.7,1) {\includegraphics[width=0.3\textwidth]{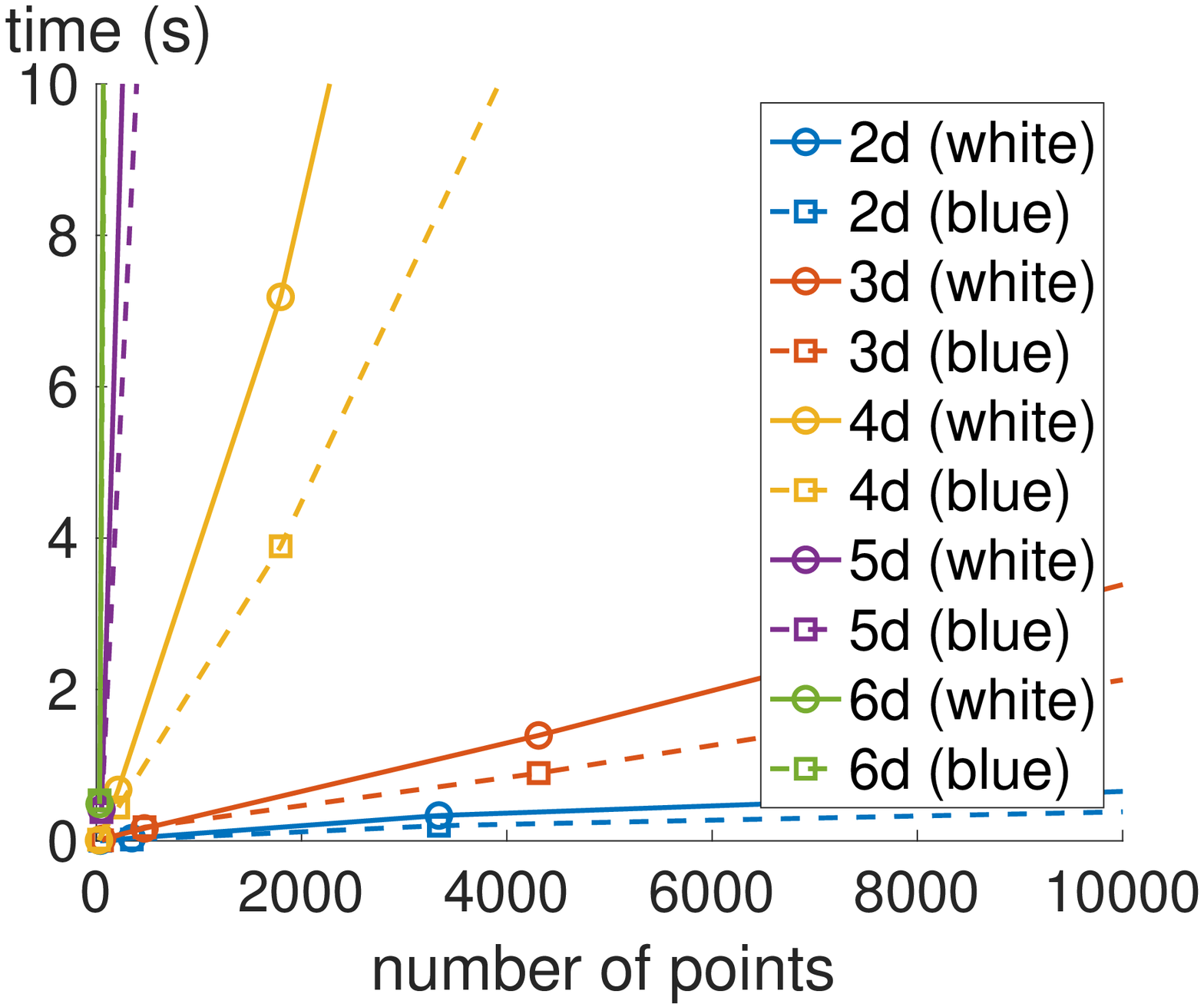}};
		\draw [line width=2pt,->,>=latex] (-1,-0.5) [out=150] to [in=60] (-2.5,-1.5);
		\draw [line width=2pt,red] (-4,-3) rectangle (-2,-2.5);
	\end{tikzpicture}
	}
	\caption{Time to compute the Voronoi diagram in $2d$-$6d$ for white and blue noise distributions with a varying number of points (Voronoi sites).}
	\label{fig:cost2d-6d}
\end{figure}

\begin{figure}[h!]
	\centering
	\includegraphics[width=0.4\textwidth]{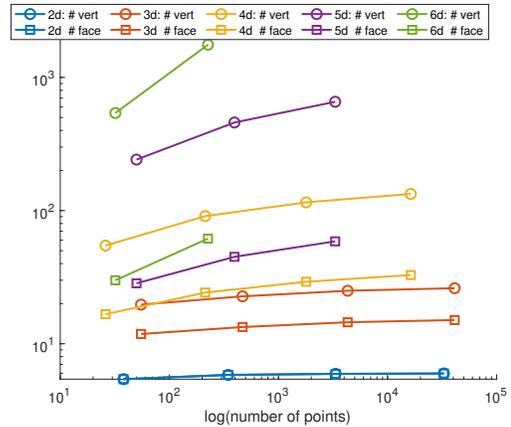}
	\caption{Total number of vertices (circles) and facets (squares) for Voronoi diagrams computed from a random point distribution in $2d-6d$ for a varying number of input points (Voronoi sites).}
	\label{fig:vertices-facets-2d-6d}
\end{figure}

In our second performance study, we focus on dimensions $2$-$4$ and compute (1) the time to clip the Voronoi cells (as in the previous study), (2) the total time devoted to compute nearest neighbors, (3) the time devoted to decomposing the polytopes into integration simplices and (4) the time to perform the numerical integration for various quadrature orders $q$.
The results are shown in Table~\ref{tab:performance-2d-4d}.
In $2d$ (polygons) and $3d$ (polyhedra), the cost of the Voronoi diagram calculation outweighs any other calculation, but is still relatively fast for one million points in $2d$ (25 seconds).
The cost of decomposing the Voronoi cells into simplices is roughly one fifth the cost of computing the Voronoi diagram in $3d$ and $4d$.
In higher dimensions, the cost of integration is limiting and begins to exceed the cost of computing the Voronoi diagram, even for low quadrature orders ($q > 2$).
In general, the cost of computing the Voronoi diagram is even lower than the results shown in Table~\ref{tab:performance-2d-4d} (which is for white noise distributions) once the point distributions exhibit more structure.

\begin{table}[h]
	\centering
	\caption{Performance statistics for white noise in $2d-4d$ for a varying number of Voronoi sites ($N$).
	The timing results (in seconds) are broken into (1) the time to compute the Voronoi diagram ($t_{\mathrm{vor}}$), (2) the time devoted to computing nearest neighbors ($t_{\mathrm{knn}}$), (3) the time to decompose the $d$-polytopes into $d$-simplices ($t_{\mathrm{tri}}$) and (4) the time to perform the numerical integration a particular quadrature order $q$ ($t_q$).}
	\label{tab:performance-2d-4d}
	\resizebox{0.45\textwidth}{!} {
	\begin{tabular}{clllllll}
		\hline
		\rowcolor{gray!20}
		$d$ & $N$ & $t_{\mathrm{vor}}$ & $t_{\mathrm{knn}}$ & $t_{\mathrm{tri}}$ & $t_{q = 2}$ & $t_{q = 3}$ & $t_{q = 4}$ \\ \hline
		$2$ & $10$k & $0.037$ & $0.085$ & $0.065$ & $0.029$ & $0.039$ & $0.049$ \\
		$2$ & $100$k & $2.7$ & $0.71$ & $0.46$ & $0.17$ & $0.26$ & $0.40$ \\
		$2$ & $1$M & $25$ & $12$ & $4.2$ & $1.3$ & $2.2$ & $3.5$ \\ \hline
		$3$ & $10$k  & $2.8$ & $0.10$ & $0.41$ & $0.35$ & $1.0$ & $2.2$ \\
		$3$ & $100$k & $29$ & $1.5$  & $4.4$ & $3.4$ & $10$ & $23$ \\
		$3$ & $250$k & $71 $ & $5.2$  & $11$ & $8.6$ & $25$ & $59$ \\ \hline
		$4$ & $1$k & $4.0$ & $0.027$ & $0.82$ & $1.6$ & $7.3$ & $23$ \\
		$4$ & $10$k & $50$ & $0.40$ & $9.8$  & $18$ & $87$ & $270$ \\
		$4$ & $15$k & $78$ & $0.68$ & $15$ & $28$ & $130$ & $420$ \\ \hline
	\end{tabular}
	}
\end{table}

\section{Applications}\label{sec:applications}
We now use our power diagrams to solve quantization and semi-discrete optimal transport problems in four-dimensional domains.
The difference between the two problems is that (1) in quantization, we optimize Voronoi site coordinates and (2) in optimal transport, we optimize the weights on Voronoi sites.

\subsection{Quantization}
Here, our goal is to compute an optimal point distribution for some prescribed density function.
As we observed in the performance study, the cost of evaluating the integrals is prohibitive for a large number of points.
As a result, we design the densities to be integrated more accurately with lower quadrature orders.

\begin{algsimple}[h!]
	\begin{algcode}[optimizePoints$(N,\rho(\mathbf{x}))$]
		\zi \Inputs number of sites $N$, density $\rho(\mathbf{x})$
		\zi \Outputs optimized point distribution $\mathbf{Y}$
		\li $\mathbf{Y} \gets$ randomly sample $N$ points in domain $\Omega$
		\li $\texttt{iter} \gets 0$
		\li \For $\texttt{iter} = 1:\texttt{nb\_iter}$
		\li    compute RVD: $\mathrm{Vor}(\mathbf{Y}) \cap \Omega$ (Section~\ref{sec:power-cells})
		\li    compute mass and centroids in Eq.~\ref{eq:lloyd-relaxation}
		\li    \If \texttt{lloyd} \Comment{use Lloyd relaxation} \addindent
		\li      update $\mathbf{x}$ to current centroids\remindent
		\li    \Else \Comment{use L-BFGS update} \addindent
		\li      compute gradients $\mathrm{d}E/\mathrm{d}\mathbf{y}_i$ using Eq.~\ref{eq:de-dx}
		\li      perform L-BFGS update on $\mathbf{x}$
						\End
    \End
	\end{algcode}
	\caption{Optimizing a point distribution according to an input density measure, using either Lloyd relaxation or an L-BFGS update.
	In our applications, the input domain $\Omega$ is the unit $d$-cube represented as a polytopal mesh with a single element.
	}
	\label{alg:quantization}
\end{algsimple}

We consider three density measures.
The first is a uniform distribution $\rho_u(\mathbf{x}) = 1$ (the optimized Voronoi cells should appear uniform).
The second density measure is a Gaussian
\begin{equation}
	\rho_g(\mathbf{x}) = \frac{1}{\sqrt{(2\pi)^4\det\Sigma}} \exp\left( -\frac{1}{2} (\mathbf{x} - \boldsymbol{\mu})^T \Sigma^{-1}(\mathbf{x}-\boldsymbol{\mu}) \right)
\end{equation}
where $\boldsymbol{\mu} = (0.5,0.5,0.5,0.5)^T$ and $\Sigma = \mathrm{diag}(0.02,0.02,0.02,0.02)$.
In this case, we expect smaller Voronoi polytopes near $\boldsymbol{\mu}$, which then increase in distance around this mean point.
The third density measure we consider is that which describes an expanding sphere, which traces the geometry of a hypercone in $4d$.
In particular, this density measure is
\begin{equation}
	\rho_c(\mathbf{x}) = 100/(h^2 + 0.001)
\end{equation}
where $h$ is the distance to the cone defined by $\mathbf{r}(t) = r_0 + t\cdot\tan(\alpha)$ where $t$ represents the fourth coordinate in the domain (to be interpreted as time).
Please refer to Fig.~\ref{fig:hypercone-setup} for the geometry of the cone - note that $r_0 = 0.4$ and $r_1 = 0.7$.
Note that this density measure is effectively two-dimensional in an $r-t$ coordinate system.
Rotational symmetry about the $t$ axis leads to the cone geometry.
In fact, when slicing the resulting Voronoi diagram (with points uniformly distributed according to $\rho_c(\mathbf{x})$), we should expect to see smaller cells clustering around the geometry of a $3d$ cone when slicing along a non-constant $t$ hyperplane, but should expect to see smaller cells clustering around a $3d$ sphere at constant $t$ hyperplanes.
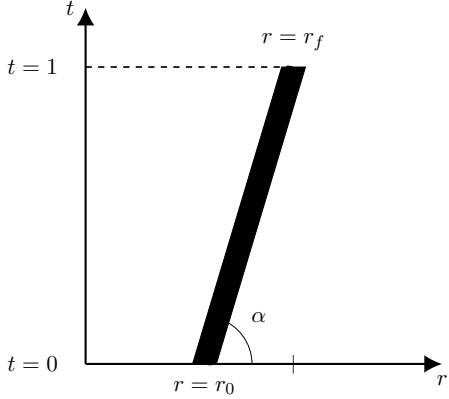
\begin{figure}
	\centering
	\resizebox{0.375\textwidth}{!} {
		\begin{tikzpicture}
			\draw[-{Latex[length=3mm, width=3mm]},line width=1pt] (0,0) -- (6,0);
			\draw[-{Latex[length=3mm, width=3mm]},line width=1pt] (0,0) -- (0,6);
			\node[below = 2pt] at (6,0) {\large$r$};
			\node[left = 2pt] at (0,6) {\large$t$};

			\draw[thick,line width=6pt] (2,0) -- (3.5,5);
			\draw[fill=black] (1.8,0) -- (2.2,0) -- (3.7,5) -- (3.3,5) -- (1.8,0);

			\draw[thick,dashed] (0,5) -- (3.5,5);

			\node[left=10pt] at (0,0) {\large$t = 0$};
			\node[left=10pt] at (0,5) {\large$t = 1$};

			\node[below=5pt] at (2,0) {\large$r = r_0$};
			\node[above=5pt] at (3.5,5) {\large$r = r_f$};

			\node [above=8pt,right=5pt] at (2.5,0.5) {\large$\alpha$};

			\draw (2,0) ++(73.3:.8) arc (73.3:0:.8);

			\node at (3.5,0) {$|$};

		\end{tikzpicture}
	}
		\caption{Definition of the hypercone in a radial-temporal coordinate system.
		The sphere starts with a radius of $r_0 = 0.4$ at $t = 0$ and expands to a radius of $r_1 = 0.7$ at $t = 1$.
		}
		\label{fig:hypercone-setup}
\end{figure}
\begin{figure}[h]
	\centering
	\includegraphics[width=0.35\textwidth]{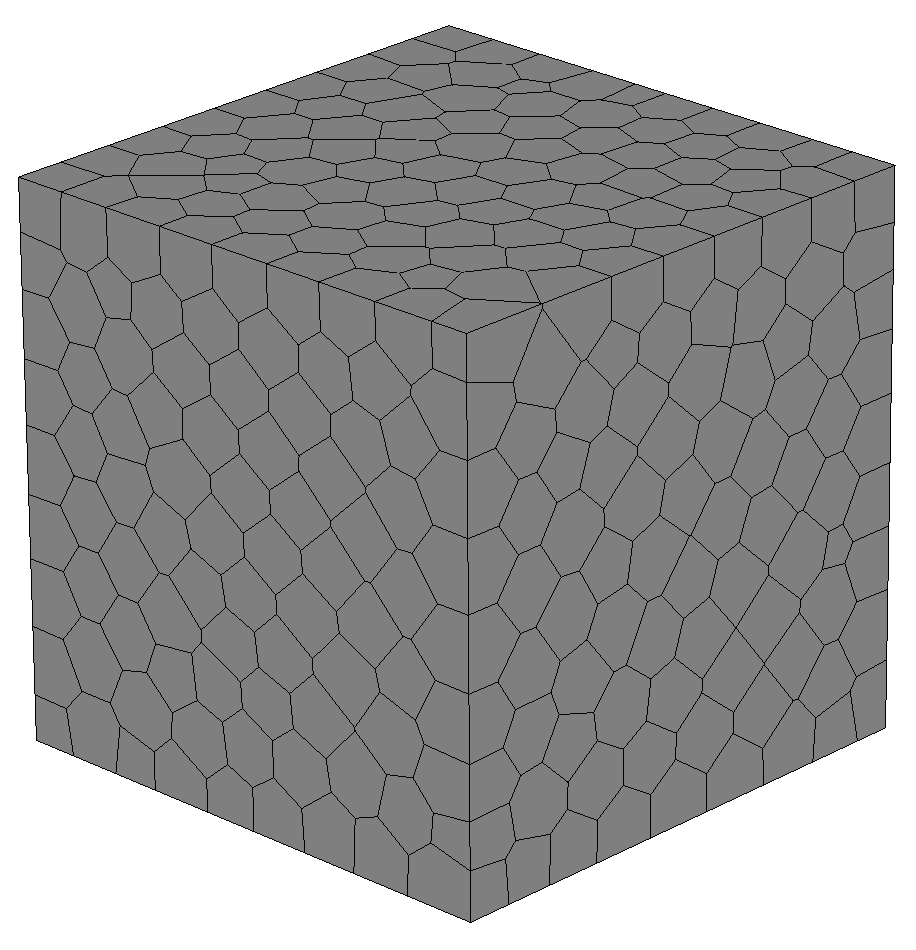}
	\caption{Optimized four-dimensional Voronoi diagram sliced at $t = 0$ (where $t$ is the fourth dimension) for the case of a uniform density, $\rho_u(\mathbf{x})$.}
	\label{fig:qntz-uniform}
\end{figure}%

We optimize four-dimensional point distributions with $N = 10,000$ points using Lloyd relaxation~\cite{Lloyd_1982} the gradient-based L-BFGS optimizer~\cite{Johnson_NLOPT} - this procedure is outlined in Algorithm~\ref{alg:quantization}.
After 100 iterations (function calls in the \texttt{nlopt} optimizer), the point distributions indeed exhibit a uniform distribution under the prescribed density measure.
Fig.~\ref{fig:qntz-uniform} shows uniformly distributed polytopes sliced along the $t = 0$ hyperplane.
Fig.~\ref{fig:qntz-gaussian} shows the interior of a slice along the $t = 0.5$ hyperplane for the Gaussian density.
The clustering of Voronoi cells around the mean is clear and the cell sizes increases with distance from the mean.
Furthermore, the expected clustering around the hypercone is visible in Fig.~\ref{fig:qntz-cone}.
Specifically, when the four-dimensional Voronoi diagram is sliced along a hyperplane with non-constant $t$ ($x = 0$), we can see clustering near a $3d$ cone.
We can also see smaller polytopes clustered around a sphere when the optimized Voronoi diagram is sliced at $t = 0$ (smaller sphere with radius $r_0$) and $t = 1$ (larger sphere with radius $r_1$).
\begin{figure}[h!]
	\centering
  \begin{subfigure}{0.4\textwidth}
		\includegraphics[width=\textwidth]{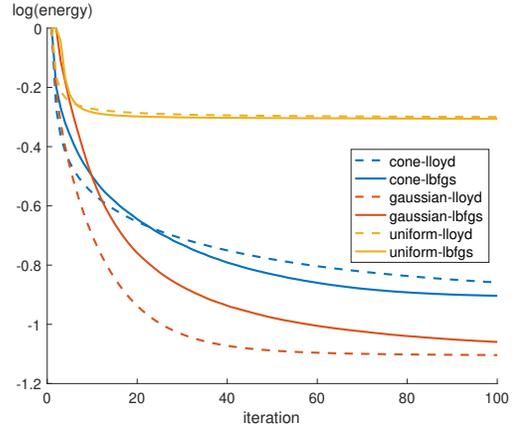}
		\caption{Energy (Eq.~\ref{eq:energy-functional} versus iteration.}
	\end{subfigure}
	\begin{subfigure}{0.4\textwidth}
		\includegraphics[width=\textwidth]{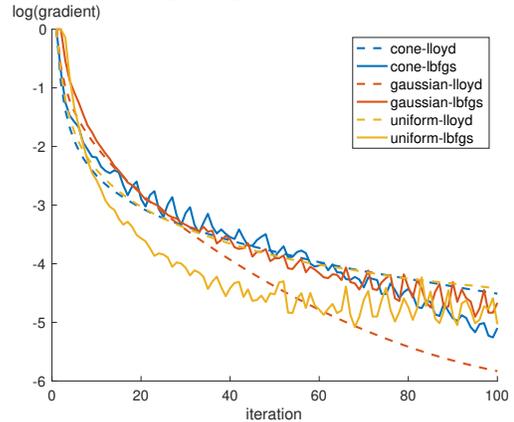}
		\caption{Gradient norm (of Eq.~\ref{eq:de-dx}) versus iteration.}
	\end{subfigure}
  \caption{Convergence of the energy and its gradient for Lloyd relaxation and LBFGS optimizer for all three density functions (uniform, Gaussian, cone) considered in the quantization application.}
	\label{fig:qntz-convergence}
\end{figure}%
Fig.~\ref{fig:qntz-convergence} demonstrates the convergence of the energy functional and gradient norm during the optimization for both Lloyd relaxation (even though it is not driven by the gradient) and the L-BFGS optimizer.
The data in each curve is normalized by the initial energy and gradient norm at the onset of the optimization.
The L-BFGS optimizer achieves a lower energy and gradient norm for the uniform and cone density measures, but performs slightly worse than Lloyd relaxation for the Gaussian density.
Nonetheless, the results demonstrate the ability to obtain point distributions that are uniform with respect to a prescribed density function in $4d$ for a relatively large number of points ($N = 10,000$).
\begin{figure}[h]
	\centering
	\begin{subfigure}{0.325\textwidth}
		\includegraphics[width=\textwidth]{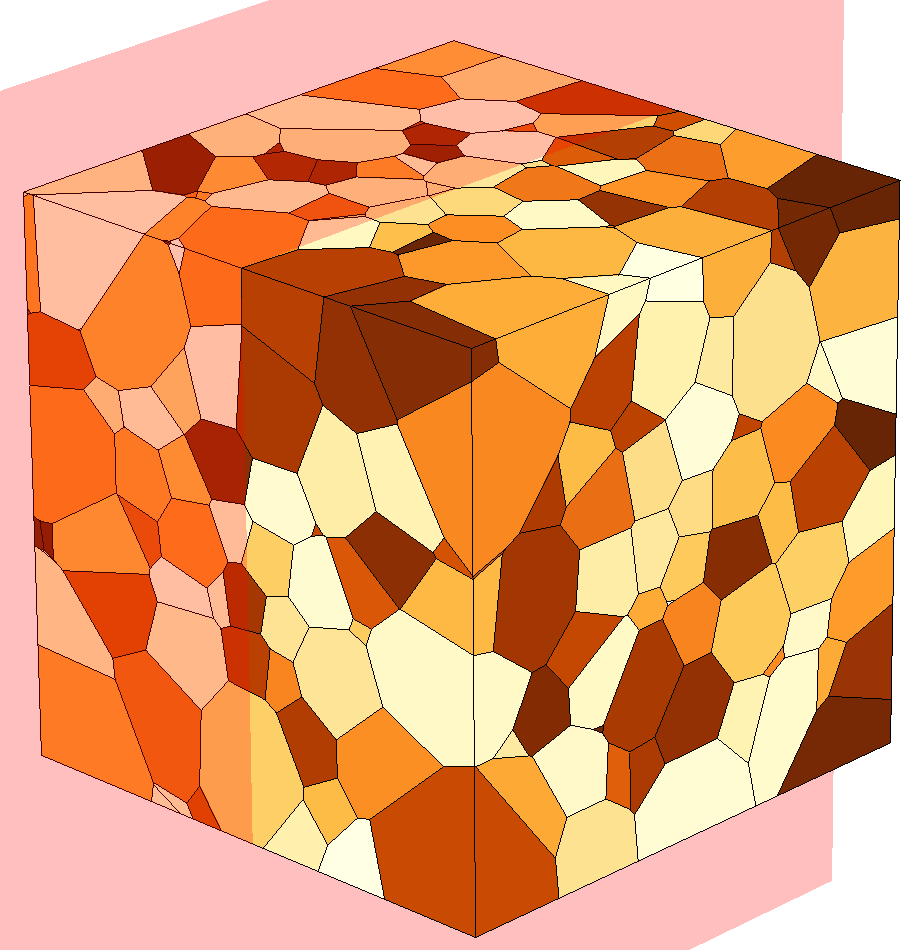}
		\caption{Optimized Voronoi diagram sliced at $t = 0.5$ with clipping plane shown.}
	\end{subfigure}
	\begin{subfigure}{0.325\textwidth}
			\includegraphics[width=\textwidth]{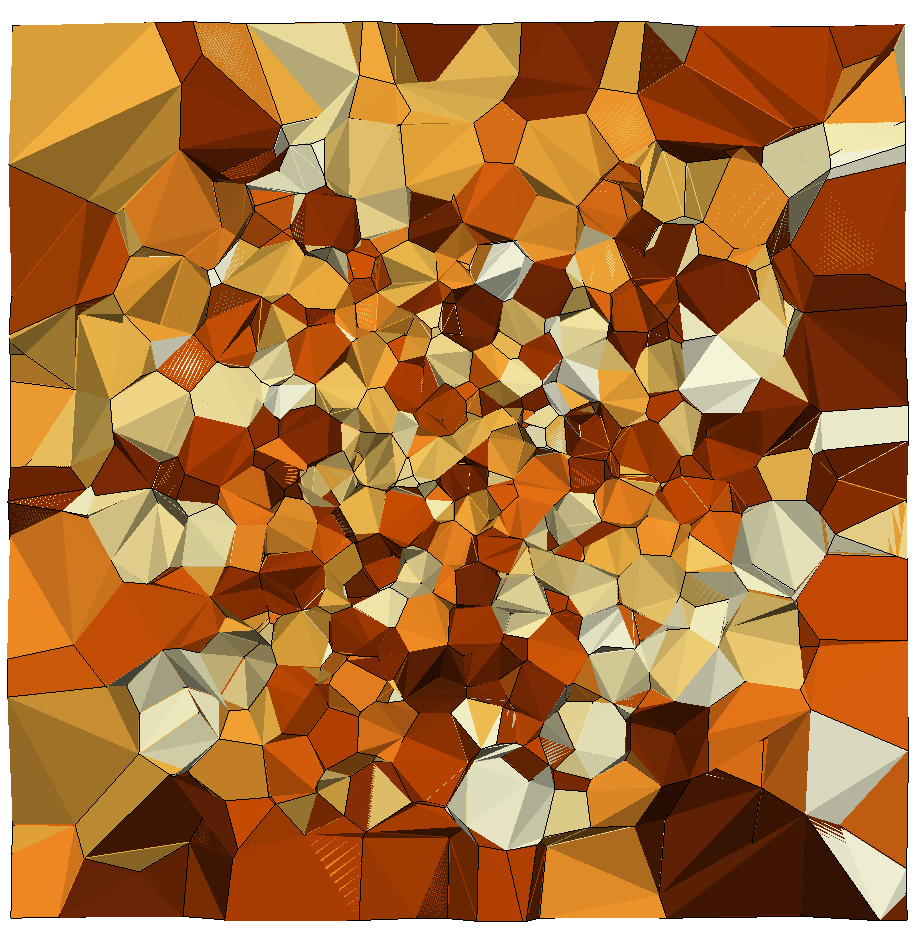}
			\caption{Optimized Voronoi diagram sliced at $t = 0.5$, clipped at $x = 0.5$ (transparent clipping plane shown in pink above).}
	\end{subfigure}
	\caption{Optimized four-dimensional Voronoi diagram sliced at $t = 0.5$ (where $t$ is the fourth dimension) for the case of a Gaussian density, $\rho_g(\mathbf{x})$.
	The bottom figure shows a slice (at $x = 0.5$) of the top figure, where the clipping plane is defined by the transparent plane shaded in pink.
	}
	\label{fig:qntz-gaussian}
\end{figure}%
\begin{figure}[h!]
	\centering
	\begin{subfigure}{0.295\textwidth}
			\includegraphics[width=\textwidth]{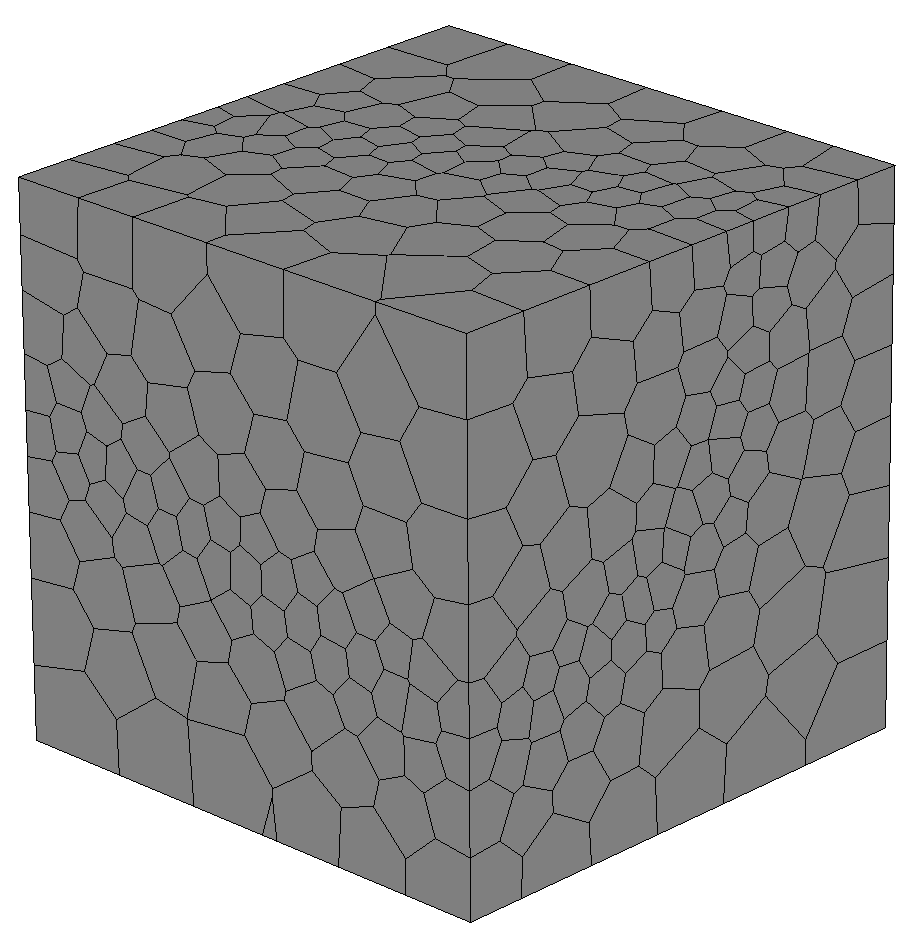}
			\caption{Voronoi diagram slice at $x = 0$.}
	\end{subfigure}
	\begin{subfigure}{0.295\textwidth}
			\includegraphics[width=\textwidth]{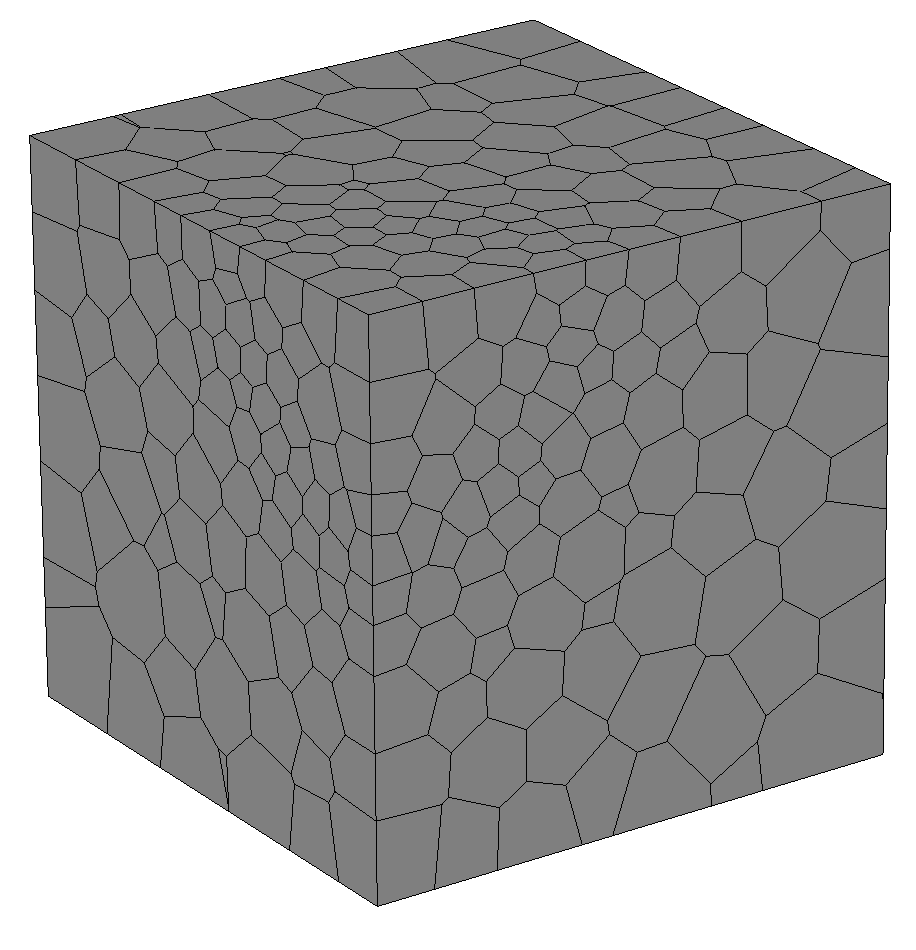}
			\caption{Voronoi diagram slice at $t = 0$.}
	\end{subfigure}
	\begin{subfigure}{0.295\textwidth}
			\includegraphics[width=\textwidth]{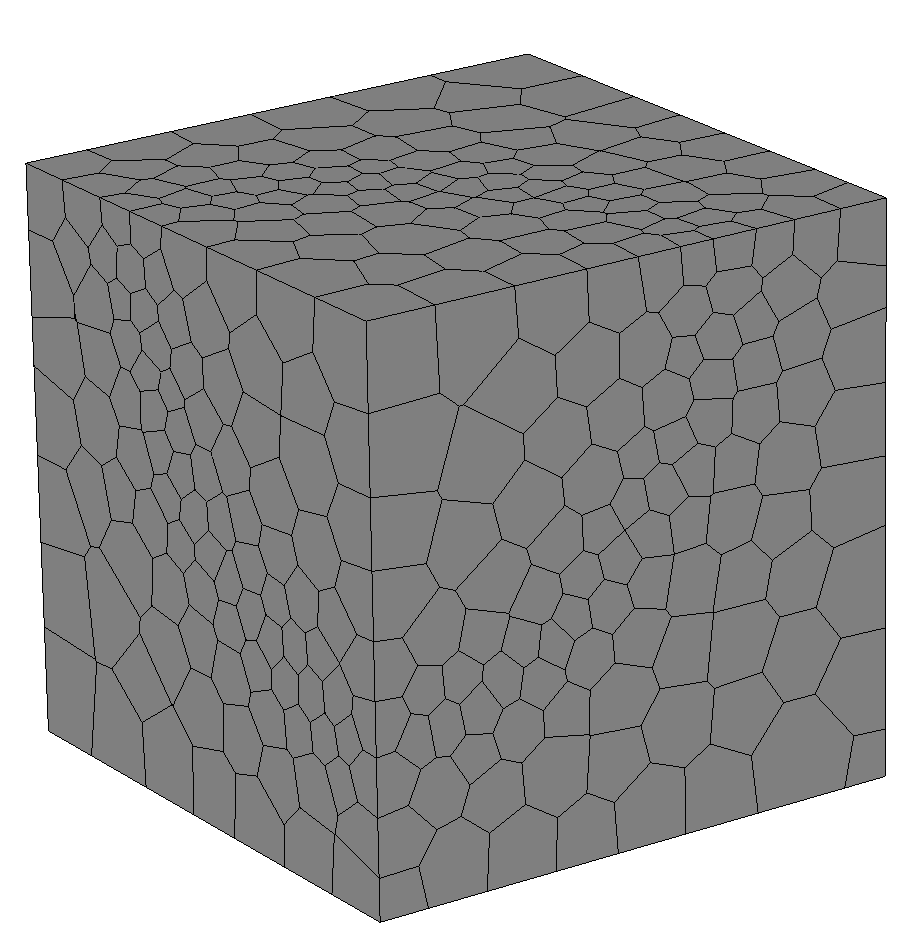}
			\caption{Voronoi diagram slice at $t = 1$.}
	\end{subfigure}
	\caption{Optimized four-dimensional Voronoi diagram sliced at $t = 0$, $t = 1$ (where $t$ is the fourth dimension) and $x = 0$ for the case of a cone-like density, $\rho_c(\mathbf{x})$.
	A three-dimensional cone is visible at $x = 0$ (where the smaller cells cluster), whereas the sphere with initial radius $r_0$ and final radius $r_1$ is shown at $t = 0$ and $t = 1$, respectively.
	Only one quarter of the cone and spheres are visible within this domain.
	}
	\label{fig:qntz-cone}
\end{figure}%
\begin{algsimple}[h!]
	\begin{algcode}[optimizeWeights$(N,\rho(\mathbf{x}))$]
		\zi \Inputs number of sites $N$, density $\rho(\mathbf{x})$
		\zi \Outputs power cells with uniform mass
		\li $\mathbf{Y} \gets \Call{optimizePoints}(N,\rho)$
		\li $\mathbf{w} \gets \mathbf{0}$
		\li $\nu_t \gets m_t/N$\Comment{$m_t$ is the total mass}
		\li \For $\texttt{iter} = 1:\texttt{nb\_iter}$
		\li 	 $\mathbf{Z} \gets $lift $\mathbf{Y}$ to $\mathbb{R}^{d+1}$ (Eq.~\ref{eq:lift-sites})
		\li    compute $\mathrm{Vor}(\mathbf{Z}) \cap \Omega$ (Section~\ref{sec:power-cells})
		\li    compute energy and mass in Eq.~\ref{eq:lloyd-relaxation}
		\li    compute gradients $\mathrm{d}E/\mathrm{d}{w}_i$ using Eq.~\ref{eq:de-dw}
		\li    perform L-BFGS update on $\mathbf{w}$
		\zi    \addindent \hspace{20pt}with $-E$ and $-\mathrm{d}E/\mathrm{d}\mathbf{w}$\remindent
    \End
	\end{algcode}
	\label{alg:sdot}
	\caption{Optimizing the weights to achieve a target mass according to an input density measure.
		In our applications, the input domain $\Omega$ is the unit $d$-cube represented as a polytopal mesh with a single element.
	}
	\label{alg:sdot}
\end{algsimple}
\subsection{Semi-discrete optimal transport}
We now turn our attention to the problem of assigning an equal mass (measured under some input density function) to each Voronoi cell.
We will study smaller problem sizes because we do not currently implement a multiscale algorithm to initialize the weights for the next ``level," although this would certainly accelerate the convergence of the optimization as pointed out by M\'erigot~\cite{Merigot_2011,Merigot_2017} and L\'evy~\cite{Levy_2015}.
We thus use $N = 1000$ points distributed within a four-dimensional domain.

The two density measures we consider here are (1) a uniform density $\rho_u(\mathbf{x}) = 1$ and (2) a spherical density $\rho_s(\mathbf{x}) = 1 + 100\,\lvert|\mathbf{x} - \boldsymbol{\mu}\rvert|^2$ where $\boldsymbol{\mu} = (0.5,0.5,0.5,0.5)^T$.
In contrast to the previous section which employed a Gaussian as the density measure, this spherical density is more accurately integrated with lower quadrature orders, thus keeping the computational cost reasonable for this demonstration.

Before optimizing the weights, we first optimize the point distributions with the L-BFGS method (of the previous section) to achieve a uniform distribution with respect to the input density measure - see Algorithm~\ref{alg:sdot}.
We then obtain the total mass $m_{t}$ by integrating the input density over the entire domain and set the target mass $\nu_t = m_{t}/N$ for every site.
Finally, we optimize the weights by iteratively computing the power diagram and integrating the quantities in Eq.~\ref{eq:de-dw} which are then passed to the L-BFGS optimizer to determine the next set of weights.

The convergence of the energy functional and gradient norm (Eq.~\ref{eq:de-dw}) is shown in Fig.~\ref{fig:sdot-energy}.
Observe that the optimization procedure exhibits a very fast convergence near the $75$th iteration.
Furthermore, Fig.~\ref{fig:sdot-mass} shows the distribution of mass at each iteration of the optimization, normalized by the target mass $\nu_t$ - thus we strive for a normalized mass of $1$ for each cell.
Initially, the cell masses are distributed across a wide range and converge to the target mass at about the $75$th iteration.
This demonstrates the ability of our algorithm to achieve a target mass for each Voronoi cell.
\begin{figure}[h!]
	\centering
  \begin{subfigure}{0.4\textwidth}
		\includegraphics[width=\textwidth]{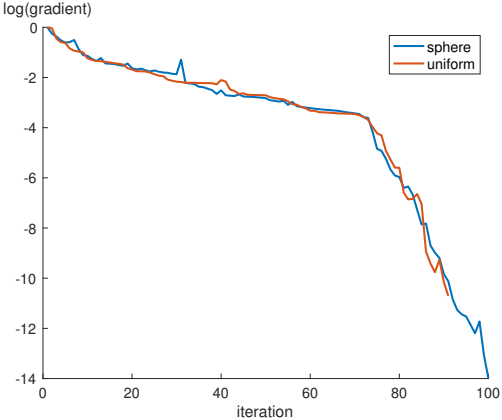}
		\caption{Gradient norm (of Eq.~\ref{eq:de-dw}) versus iteration.}
		\label{fig:sdot-energy}
	\end{subfigure}
	\begin{subfigure}{0.4\textwidth}
		\includegraphics[width=\textwidth]{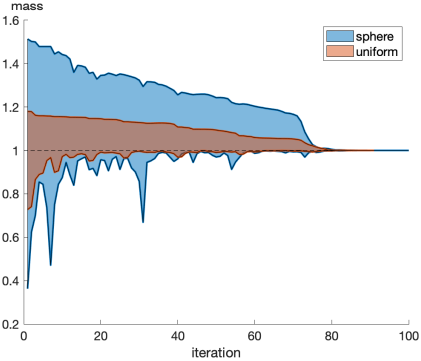}
		\caption{Distribution of power cell mass versus iteration.}
		\label{fig:sdot-mass}
	\end{subfigure}
	\caption{Convergence of the energy gradient and normalized mass for all the uniform and sphere-like density functions considered in the semi-discrete optimal transport application.
	The dashed line in Fig.~\ref{fig:sdot-mass} represents the target mass.
	}
	\label{fig:sdot-convergence}
\end{figure}
\section{Conclusions \& Future Work}
In this paper, we have presented the theory and implementation details for computing higher-dimensional power diagrams, and demonstrated the first implementation for solving four-dimensional quantization and $L^2$ semi-discrete optimal transport problems which seeks a uniform target mass of each Voronoi cell, measured under some input density function.

The performance of the algorithm was also evaluated in up to six dimensions (the topological dimension of the polytopes), which demonstrated a reasonable scaling with the number of input points in dimensions $2-4$.
The cost of computing the Voronoi diagram with a randomly distributed point set was higher than that to compute the Voronoi diagram with a point distribution that exhibited blue noise properties.
Also, the cost of performing the numerical integration to compute gradients significantly outweighed the cost of computing the power diagram, especially for four-dimensional polytopes.
Future work may consist of deriving less expensive quadrature schemes in higher dimensions, which is under way in the works of Williams~\cite{Williams_2020} and Frontin~\cite{Frontin_2020}.

Furthermore, the cost of computing the power diagram is higher than to compute the Voronoi diagram because the introduction of weights lifts the points to a unit codimensional Euclidean space, thereby increasing the distance between sites and requiring several more (possibly non-contributing) bisectors to be clipped against.
This is a failure case of the security radius theorem, as pointed out by Sainlot et al.~\cite{Sainlot_2017}, therefore, future work could consist of extending their corner validation algorithm to the higher-dimensional setting.

Our algorithm is also very well-suited for an implementation on the GPU, similar to the work of Ray et al.~\cite{Ray_2019} which demonstrated an efficient $3d$ restricted Voronoi diagram calculation on the GPU.
Our algorithm could also be made more efficient by employing a multiscale approach~\cite{Merigot_2011,Merigot_2017} in the optimization of the weights.
It would also be useful to implement a more robust optimization which avoids weights that cause power cells to vanish.

Finally, our four-dimensional power diagrams could be used to compute the transport distance between a data set and a continuous density measure in both space and time - Janati et al. recently consider a fully discrete version of this problem~\cite{Janati_2020}.
Other applications could include performing coupled space-time numerical simulations of physical phenomena.
\FloatBarrier
\bibliography{references}

\begin{thebibliography}{10}
\newcommand{\enquote}[1]{``#1''}
\expandafter\ifx\csname url\endcsname\relax
  \def\url#1{\texttt{#1}}\fi
\expandafter\ifx\csname urlprefix\endcsname\relax\def\urlprefix{URL }\fi

\bibitem{Monge_1781}
Monge G.
\newblock \enquote{M\'emoire sur la th\'eorie des d\'eblais et des remblais.}

\bibitem{Balzer_2009}
Balzer M., Schl\"omer T., Deussen O.
\newblock \enquote{Capacity-constrained point distributions: a variant of
  Lloyd's method.}
\newblock \emph{ACM Transasctions on Graphics}, vol.~28, no.~3, 2009

\bibitem{deGoes_2012}
de~Goes F., Breeden K., Ostromoukhov V., Desbrun M.
\newblock \enquote{Blue noise through optimal transport.}
\newblock \emph{ACM Transactions on Graphics}, vol.~31, no.~6, 1--11, 2012

\bibitem{Ma_2018a}
Ma L., Chen Y., Qian Y., Sun H.
\newblock \enquote{Incremental Voronoi sets for instant stippling.}
\newblock \emph{The Visual Computer}, vol.~34, 863--873, 2018

\bibitem{Ma_2018b}
Ma L., Guo J., Yan D.M., Sun H., Chen Y.
\newblock \enquote{Instant stippling on $3d$ scenes.}
\newblock \emph{Pacific Graphics}, vol.~37. 2018

\bibitem{Meyron_2018}
Meyron J.
\newblock \emph{Semi-discrete optimal transport and applications in non-imaging
  optics}.
\newblock Ph.D. thesis, 2018

\bibitem{Meyron_2019}
Meyron J.
\newblock \enquote{Initialization procedures for discrete and semi-discrete
  optimal transport.}
\newblock \emph{Computer-Aided Design}, vol. 115, 13--22, 2019

\bibitem{Peyre_2019}
Peyré G., Cuturi M.
\newblock \enquote{Computational Optimal Transport: With Applications to Data
  Science.}
\newblock \emph{Foundations and Trends® in Machine Learning}, vol.~11, no.
  5-6, 355--607, 2019

\bibitem{deGoes_2011}
de~Goes F., Cohen-Steiner D., Allez P., Desbrun M.
\newblock \enquote{An optimal transport approach to robust reconstruction and
  simplification of $2d$ shapes.}
\newblock \emph{Computer Graphics Forum}, vol.~30, no.~5, 1593--1602, 2011

\bibitem{Bonneel_2019}
Bonneel N., Coeurjolly D.
\newblock \enquote{SPOT: sliced partial optimal transport.}
\newblock \emph{ACM Transactions on Graphics}, vol.~38, no.~4, 1--13, 2019

\bibitem{Paulin_2020}
Paulin L., Bonneel N., Coeurjolly D., Iehl J.C., Webanck A., Desbrun M.,
  Ostromoukhov V.
\newblock \enquote{Sliced optimal transport sampling.}
\newblock \emph{ACM Transactions on Graphics}, vol.~39, no.~4, 2019

\bibitem{Hartmann_2020}
Hartmann V., Schuhmacher D.
\newblock \enquote{Semi-discrete optimal transport: a solution procedure for
  the unsquared Euclidean distance.}
\newblock \emph{Mathematical Methods of Operations Research}, vol.~92,
  133--163, 2020

\bibitem{Levy_2018}
L\'evy B., Schmidt E.L.
\newblock \enquote{Notions of optimal transport and how to implement them on a
  computer.}
\newblock \emph{Computers and Graphics}, vol.~72, 135--148, 2018

\bibitem{deGoes_2015}
de~Goes F., Wallez C., Huang J., Pavlov D., Desbrun M.
\newblock \enquote{Power particles: an incompressible fluid solver based on
  power diagrams.}
\newblock \emph{ACM Transactions on Graphics}, vol.~34, no.~4, 1--11, 2015

\bibitem{Benamou_2000}
Benamou J.D., Brenier Y.
\newblock \enquote{A computational fluid mechanics solution to the
  Monge-Kantorovich mass transfer problem.}
\newblock \emph{Numerische Mathematik}, vol.~84, no.~3, 375--393, 2000

\bibitem{Villani_2009}
Villani C.
\newblock \emph{Optimal transport: old and new}.
\newblock Springer, 2009

\bibitem{Solomon_2017}
Solomon J.
\newblock \enquote{Computational optimal transport.}
\newblock \emph{Snapshots of Modern Mathematics}, vol.~8, 2017

\bibitem{Merigot_2017}
Mérigot Q., Meyron J., Thibert B.
\newblock \enquote{An algorithm for optimal transport between a simplex soup
  and a point cloud.}
\newblock \emph{SIAM Journal on Imaging Sciences}, vol.~11, 2017

\bibitem{Cuturi_2013}
Cuturi M.
\newblock \enquote{Sinkhorn Distances: Lightspeed Computation of Optimal
  Transport.}
\newblock C.J.C. Burges, L.~Bottou, M.~Welling, Z.~Ghahramani, K.Q. Weinberger,
  editors, \emph{Advances in Neural Information Processing Systems}, vol.~26,
  pp. 2292--2300. Curran Associates, Inc., 2013

\bibitem{Aurenhammer_1987}
Aurenhammer F.
\newblock \enquote{Power diagrams: properties, algorithms and applications.}
\newblock \emph{SIAM Journal on Computing}, vol.~16, no.~1, 78--96, 1987

\bibitem{Aurenhammer_1991}
Aurenhammer F.
\newblock \enquote{Voronoi diagrams: a survey of a fundaental geometric data
  structure.}
\newblock \emph{ACM Computing Surveys}, vol.~23, no.~3, 1991

\bibitem{Canas_2006_Surface_remeshing_arbitrary_codimension}
Ca\~nas G.D., Gortler S.J.
\newblock \enquote{Surface Remeshing in Arbitrary Codimensions.}
\newblock \emph{The Visual Computer}, vol.~22, no.~9, 885--895, Sep. 2006

\bibitem{Levy_2013_Vorpaline}
L\'evy B., Bonneel N.
\newblock \enquote{Variational Anisotropic Surface Meshing with {V}oronoi
  Parallel Linear Enumeration.}
\newblock \emph{Proceedings of the 21st International Meshing Roundtable}. 2012

\bibitem{Dassi_2015}
Dassi F., Si H., Perotto S., Streckenbach T.
\newblock \enquote{Anisotropic Finite Element Mesh Adaptation via Higher
  Dimensional Embedding.}
\newblock \emph{Procedia Engineering}, vol. 124, 265 -- 277, 2015.
\newblock 24th International Meshing Roundtable

\bibitem{Dassi_2016}
Dassi F., Farrell P., Si H.
\newblock \enquote{An Anisoptropic Surface Remeshing Strategy Combining Higher
  Dimensional Embedding with Radial Basis Functions.}
\newblock \emph{Procedia Engineering}, vol. 163, 72 -- 83, 2016.
\newblock 25th International Meshing Roundtable

\bibitem{Dassi_2017}
Dassi F., Perotto S., Si H., Streckenbach T.
\newblock \enquote{A Priori Anisotropic Mesh Adaptation Driven by a Higher
  Dimensional Embedding.}
\newblock \emph{Computer-Aided Design}, vol.~85, 111 -- 122, 2017

\bibitem{Nivoliers_2015}
Nivoliers V., L{\'e}vy B., Geuzaine C.
\newblock \enquote{Anisotropic and Feature Sensitive Triangular Remeshing Using
  Normal Lifting.}
\newblock \emph{Journal of Computational and Applied Mathematics}, vol. 289,
  225--240, Dec. 2015

\bibitem{Levy_2020}
L\'evy B., Mohayaee R., von Hausegger S.
\newblock \enquote{A fast semi-discrete optimal transport algorithm for a
  unique reconstruction of the early Universe.}
\newblock 2020

\bibitem{Levy_2016_Geogram}
L\'evy B.
\newblock \enquote{{G}eogram: A Programming Library of Geometric Algorithms.},
  2016

\bibitem{Rycroft_2009}
Rycroft C.H.
\newblock \enquote{Voro++: A three-dimensional Voronoi cell library in C++.}
\newblock \emph{Chaos}, vol.~19, no.~4, 2009

\bibitem{CGAL_software}
{The CGAL Project}.
\newblock \emph{{CGAL} User and Reference Manual}.
\newblock {CGAL Editorial Board}, {5.2} edn., 2020.
\newblock \urlprefix\url{https://doc.cgal.org/5.2/Manual/packages.html}

\bibitem{CGAL_Voronoi}
Karavelas M.
\newblock \enquote{{2D} Voronoi Diagram Adaptor.}
\newblock \emph{{CGAL} User and Reference Manual}. {CGAL Editorial Board},
  {5.2} edn., 2020

\bibitem{Sutherland_1974}
Sutherland I.E., Hodgman G.W.
\newblock \enquote{Reentrant polygon clipping.}
\newblock \emph{Communications of the ACM}, vol.~17, no.~1, 1974

\bibitem{Henk_2004_Basic_properties_of_convex_polytopes}
Henk M., Richter-Gebert J., Ziegler G.M.
\newblock \enquote{Basic Properties of Convex Polytopes.}
\newblock \emph{Handbook of Discrete and Computational Geometry, 2nd Ed.} 2004

\bibitem{Ziegler_1995_Lectures_on_Polytopes}
Ziegler G.M.
\newblock \emph{Lectures on Polytopes}, vol. 152 of \emph{Graduate Texts in
  Mathematics}.
\newblock Springer, 1995

\bibitem{Levy_2015}
L\'evy B.
\newblock \enquote{A numerical algorithm for $L^2$ semi-discrete optimal
  transport.}
\newblock \emph{ESAIM: Mathematical Modelling and Numerical Analysis}, vol.~49,
  1693--1715, 2015

\bibitem{Santambrogio_2015}
Santambrogio F.
\newblock \emph{Optimal transport for applied mathematicians}.
\newblock Birkhauser Basel, 2015

\bibitem{Villani_2003}
Villani C.
\newblock \emph{Topics in optimal transportation}, vol.~58.
\newblock American Mathematical Society, 2003

\bibitem{Kantorovich_1958}
Kantorovitch L.
\newblock \enquote{On the translocation of masses.}
\newblock \emph{Management Science}, vol.~5, no.~1, 1--4, 1958

\bibitem{Lloyd_1982}
Lloyd S.P.
\newblock \enquote{Least Squares Quantization in {PCM}.}
\newblock \emph{IEEE Transaction on Information Theory}, vol.~28, no.~2,
  129--137, 1982

\bibitem{Du_1999_CVT_Algorithm_Application}
Du Q., Faber V., Gunzburger M.
\newblock \enquote{Centroidal {V}oronoi Tessellations: Applications and
  Algorithms.}
\newblock \emph{SIAM Review}, vol.~41, 637--676, 1999

\bibitem{Liu_1989}
Liu D.C., Nocedal J.
\newblock \enquote{On the limited memory BFGS method for large scale
  optimiation.}
\newblock \emph{Mathematical Programming}, vol.~45, 503--528, 1989

\bibitem{Johnson_NLOPT}
Johnson S.G.
\newblock \enquote{The {NL}opt nonlinear-optimization package, available at:
  \texttt{http://ab-initio.mit.edu/nlopt}.}

\bibitem{Kitagawa_2016}
Kitagawa J., M\'erigot Q., Thibert B.
\newblock \enquote{Convergence of a Newton algorithm for semi-discrete optimal
  transport.}
\newblock 2016

\bibitem{Merigot_2011}
M\'erigot Q.
\newblock \enquote{A multiscale approach to optimal transport.}
\newblock \emph{Computer Graphics Forum}, vol.~30, 1583--1592, 2011

\bibitem{Caplan_2019_PhD}
Caplan P.C.
\newblock \emph{Four-dimensional Anisotropic Mesh Adaptation for Spacetime
  Numerical Simulations}.
\newblock {PhD} thesis, Massachusetts Institute of Technology, 2019

\bibitem{Shewchuk_1996_Adaptive_Precision}
Shewchuk J.R.
\newblock \enquote{Adaptive Precision Floating-Point Arithmetic and Fast Robust
  Geometric Predicates.}
\newblock \emph{Discrete \& Computational Geometry}, vol.~18, no.~3, 305--363,
  1996

\bibitem{Levy_2016_PCK}
L\'evy B.
\newblock \enquote{{Robustness and Efficiency of Geometric Programs: The
  Predicate Construction Kit}.}
\newblock \emph{Computer-Aided Design}, vol.~72, 3--12, 2016

\bibitem{Stroud_1971_Approximate_calculation_of_multiple_integrals}
Stroud A.
\newblock \emph{{Approximate Calculation of Multiple Integrals}}.
\newblock Prentice-Hall Inc., 1971

\bibitem{BeldaFerrin_2020}
Ferr\'in G.B., Ruiz-Giron\'es E., Roca X.
\newblock \enquote{Visualization of pentatopic meshes.}
\newblock Tech. rep., 2020

\bibitem{Si_2005_TetGen}
Si H.
\newblock \enquote{TetGen: A Quality Tetrahedral Mesh Generator and
  Three-Dimensional {D}elaunay Triangulator.}
\newblock Weierstrass Institute for Applied Analysis and Stochastics, 2005.
\newblock \texttt{http://tetgen.berlios.de}

\bibitem{Mitchell_2018}
Mitchell S.A.
\newblock \enquote{Spoke-Darts for high-dimensional blue-noise sampling.}
\newblock \emph{ACM Transactions on Graphics}, vol.~37, no.~2, 2018

\bibitem{Williams_2020}
Williams D.M., Frontin C.V., Miller E.A., Darmofal D.L.
\newblock \enquote{A family of symmetric, optimized quadrature rule for
  pentatopes.}
\newblock \emph{Computers \& Mathematics with Applications}, vol.~80, no.~5,
  1405--1420, 2020

\bibitem{Frontin_2020}
Frontin C.V., Walters G.S., Witherden F.D., Lee C.W., Williams D.M., Darmofal
  D.L.
\newblock \enquote{Foundations of space-time finite element methods: polytopes,
  interpolation and integration.}
\newblock \emph{arXiV preprint}, 2020

\bibitem{Sainlot_2017}
Sainlot M., Nivoliers V., Attali D.
\newblock \enquote{Restricting Voronoi diagrams to meshes using corner
  validation.}
\newblock \emph{Eurographics Symposium on Geometry Processing}, vol.~36, no.~5,
  2017

\bibitem{Ray_2019}
Ray N., Sokolov D., Lefebvre S., L\'evy B.
\newblock \enquote{Meshless Voronoi on the GPU.}
\newblock \emph{ACM Transactions on Graphics}, vol.~37, no.~6, 2018

\bibitem{Janati_2020}
Janati H., Cuturi M., Gramfort A.
\newblock \enquote{Spatio-temporal alignments: optimal transport through space
  and time.}
\newblock S.~Chiappa, R.~Calandra, editors, \emph{Proceedings of the Twenty
  Third International Conference on Artificial Intelligence and Statistics},
  vol. 108, pp. 1695--1704. 26--28 Aug 2020

\end{thebibliography}

\end{document}